\documentclass[12pt]{amsart}
\usepackage{amssymb,amsmath}
\def\R {\mathbb{R}}

\def\Dt{\partial_t}

\def\supp{\operatorname{supp}}
\def\sgn{\operatorname{sgn}}
\newtheorem{proposition}{Proposition}[section]
\newtheorem{theorem}[proposition]{Theorem}
\newtheorem{corollary}[proposition]{Corollary}
\newtheorem{lemma}[proposition]{Lemma}
\theoremstyle{definition}
\newtheorem{definition}[proposition]{Definition}
\newtheorem{remark}[proposition]{Remark}
\newtheorem{example}[proposition]{Example}
\numberwithin{equation}{section}

\def \au {\rm}
\def \ti {\it}
\def \jou {\rm}
\def \bk {\it}
\def \no#1#2#3 {{\bf #1} (#3), #2.}
\def \eds#1#2#3 {#1, #2, #3.}
\def\Dx{\Delta_x}
\def\Nx{\nabla_x}
\def\<{\left<}
\def\>{\right>}
\def\Bbb{\mathbb}
\def\eb{\varepsilon}
\def\({\left(}
\def\){\right)}

\def\<{\left<}
\def\>{\right>}

\def\Cal{\mathcal}
\def\dist{\operatorname{dist}}

\title[The Cahn-Hilliard equation]
{The Cahn-Hilliard equation with singular potentials
and dynamic boundary conditions}

\author[A. Miranville and S. Zelik]
{A. Miranville$^*$ and S. Zelik$^\dag$}
\address{$^*$Universit\'e de Poitiers
\newline\indent
Laboratoire de Math\'ematiques et Applications
\newline\indent
SP2MI
\newline\indent
Boulevard Marie et Pierre Curie - T\'el\'eport 2
\newline\indent
86962 Chasseneuil Futuroscope Cedex, France}
\email{miranv@math.univ-poitiers.fr}

\address{$^\dag$University of Surrey
\newline\indent Department of Mathematics
\newline\indent Guildford, GU2 7XH, United Kingdom}
\email{S.Zelik@surrey.ac.uk}

\begin{document}



\begin{abstract}

Our aim in this paper is to study the Cahn-Hilliard equation with
singular potentials and dynamic boundary conditions. In particular, we
prove, owing to proper approximations of the singular potential and a
suitable notion of variational solutions, the existence and uniqueness
of solutions. We also discuss the separation of the solutions from the
singularities of the potential. Finally, we prove the existence of
global and exponential attractors.

\end{abstract}

\maketitle

\section{Introduction}

The Cahn-Hilliard system

\begin{equation}
\begin{cases}
\Dt u=\kappa \Dx \mu,\ \ \kappa >0,\\
\mu=-\alpha \Dx u+f(u),\ \ \alpha >0,
\end{cases}
\end{equation}

\noindent plays an essential role in materials science as it describes
important
qualitative features of two-phase systems related with phase
separation processes. This can be observed, e.g., when
a binary alloy is cooled down sufficiently. One then observes a
partial nucleation (i.e., the apparition of nucleides in the material)
or a total nucleation, the so-called spinodal decomposition: the
material quickly becomes inhomogeneous, forming a fine-grained
structure in which each of the two components appears more or less
alternatively. In a second stage, which is called coarsening, occurs
at a slower time scale and is less understood, these microstructures
coarsen. We refer the reader to, e.g., \cite{Cah}, \cite{CahH},
\cite{KoO}, \cite{L},
\cite{MPW1}, \cite{MPW2}, \cite{NoC3} and \cite{NoC4} for more details. Here, $u$~is the order parameter (it~corresponds to a
(rescaled) density of atoms) and
$\mu $~is the chemical potential. Furthermore, $f$ is a double-well potential
whose wells correspond to the phases of the material. A
thermodynamically relevant potential is the following logarithmic
(singular) potential:

\begin{equation}
  f(s)=-2\kappa_0 s+\kappa_1 \ln {{1+s}\over {1-s}},\quad s\in
(-1,1),\quad 0<\kappa_0<\kappa_1,
 \end{equation}

\noindent although such a potential is very often approximated by
regular ones (typically, $f(s)=s^3-s$). Finally, $\kappa $ is the
mobility and $\alpha $ is related to the surface tension at the interface.

This system, endowed with Neumann boundary
conditions for both $u$ and $\mu$ (meaning that the interface
is orthogonal to the boundary and that there is no mass flux at the
boundary) or with periodic boundary conditions, has been extensively
studied and one now has a rather complete picture as far as the existence, uniqueness and regularity of
solutions and the asymptotic behavior of the solutions are
concerned. We refer the reader, among a vast literature, to, e.g.,
\cite{AW}, \cite{DeD}, \cite{Ell},
\cite{EllGar}, \cite{EllL}, \cite{EllSh}, \cite{KNP},
\cite{LiZ}, \cite{MZ1}, \cite{NScT}, \cite{NoC1}, \cite{NoC2},
\cite{NoC3}, \cite{NoC4}, \cite{HR}, \cite{WZ} and~\cite{Z}.

Now, the question of how the process of phase separation (that is, the
spinodal decomposition) is influenced by the presence of walls has
gained much attention recently (see \cite{FiMD1}, \cite{FiMD2},
\cite{KEMRSBD} and the references therein). This problem has mainly
been studied for polymer mixtures (although it should also be
important in other systems, such as binary metallic alloys): from a
technological point of view, binary polymer mixtures are particularly
interesting, since the occurring structures during the phase separation
process may be frozen by a rapid quench into the glassy state;
micro-structures at surfaces on very small length scales can be
produced in this way.

In that case, we again write that there is no mass flux at the boundary. Then,
in order to obtain the second boundary condition, following the phenomenological derivation of the Cahn-Hilliard system,
we consider, in addition to the usual Ginzburg-Landau
free energy

\begin{equation}
  \Psi_{GL}(u,\nabla u)
  = \int_\Omega  (\frac \alpha 2 |\nabla _x u|^2 + F(u) )\, dx ,
\end{equation}

\noindent where $F'=f$ and $\Omega $ is the domain occupied by the material (the chemical potential $\mu $ is defined as a variational
derivative of $\Psi_{GL}$ with respect to~$u$), and assuming that the
interactions with the walls are short-ranged, a
surface free energy of the form
\begin{equation}
  \Psi_{\Gamma }(u,\nabla _\Gamma u)
  = \int _\Gamma  (\frac {\alpha _\Gamma }2 |\nabla_\Gamma u|^2 + G(u) )\, dS,
  \ \  \alpha _\Gamma  > 0
\end{equation}

\noindent (thus, $\Psi=\Psi_{GL}+\Psi_{\Gamma }$ is the total free energy of the system),
where $\Gamma$ is the boundary of~$\Omega$ and
$\nabla_\Gamma $ is the surface gradient.
Writing finally that the system tends to minimize the excess surface
energy, we end up with the following boundary condition:

\begin{equation}
  {1\over d}\Dt u-\alpha _\Gamma \Delta _\Gamma u +g(u)+\alpha
  \partial _n u=0,\
\ {\rm on}\ \Gamma,
\end{equation}

\noindent where $\Delta _\Gamma $ is the Laplace-Beltrami operator, $\partial _n$ is
the normal derivative, $g=G'$ and $d>0$ is some relaxation parameter, which is usually
referred to as dynamic boundary condition, in the sense that the kinetics, i.e., $\Dt u$,
appears explicitly. Furthermore, in the original derivation, one has
$G(u)={1\over 2}a_\Gamma u^2-b_\Gamma u$, where $a_\Gamma >0$ accounts for a
modification of the effective interaction between the components at
the walls and $b_\Gamma $ characterizes the possible preferential
attraction (or repulsion) of one of the components by the walls (when
$b_\Gamma $ vanishes, there is no preferential attraction). We also
refer the reader to \cite{BF} and \cite{FRDGMMR} for other physical
derivations of such dynamic boundary conditions, obtained by taking the
continuum limit of lattice models within a direct mean-field
approximation and by applying a density functional theory, to
\cite{QWS} for the derivation of dynamic boundary conditions in the
context of two-phase fluids flows and to \cite{SV} and \cite{S} for an
approach based on concentrated capacity.

The Cahn-Hilliard system,
endowed with {dynamic boundary conditions},
has been studied in \cite{CFP}, \cite{Gal}, \cite{MZ2}, \cite{PRZ},
\cite{RZ} and~\cite{WZ} for regular potentials $f$ and $g$.
In particular, one now has satisfactory results on the existence, uniqueness
and regularity of solutions and on the asymptotic behavior of the solutions.

The case of nonregular potentials and dynamic boundary
conditions is essentially more complicated and less understood.
Indeed, to the best of our knowledge, even the existence of weak
energy solutions has only recently been established in that case, under the
additional restriction that the boundary nonlinearity $g$ has the
right sign at the singular points $\pm1$, namely,
\begin{equation}\label{0.sign}
\pm g(\pm1)>0
\end{equation}
(see \cite{GMS}; see also \cite{CM} where sign conditions are
considered in the context of the Caginalp phase-field system). Furthermore, the questions related with the longtime
behavior of the solutions (e.g., in terms of global attractors or/and exponential
attractors) have not
been considered in the literature.

The aim of the present paper is to give a thorough study of the
singular Cahn-Hilliard problem endowed with dynamic boundary
conditions. As we will see below,
the main difficulty here lies in the fact that the
combination of  dynamics boundary conditions and of singular
potentials can produce additional
strong singularities on the corresponding solutions
close to the boundary (especially in the case where the sign condition \eqref{0.sign} is violated).
In that case, even the simplest 1D stationary problems   may
not have solutions in a usual (or distribution) sense (due to the jumps of the normal
derivatives close to the boundary produced by the singularities, see
Example 6.2).

Nevertheless,  we can construct a sequence of
solutions of regular approximations of our
singular problem which converges to a unique trajectory which
is then naturally identified with the "solution" of the limit singular
problem. As already pointed out, this trajectory may not be a
solution of our equations in the usual (distribution) sense, so that
the notion of a solution must be properly modified. To do so, we
consider, in the spirit of
\cite{Br} (see also \cite{EGZ}), the variational {\it inequality}
associated with the problem and define a (variational) solution in
terms of this variational inequality, see Section 3 for details.

 Of course, important questions are when the solution thus defined
 is a usual distribution solution of the
equations and which additional regularity one can expect from such a variational solution.
 Actually, we prove that the variational solutions are always H\"older continuous in space and
 are solutions in the usual sense if they  do not
reach the pure states on the boundary, namely, if
\begin{equation}\label{0.reg}
\vert u(t,x)\vert<1
\end{equation}
for almost all $(t,x)\in \R^+\times \Gamma$. One possible
condition which guarantees that condition \eqref{0.reg} holds is
exactly the aforementioned sign condition \eqref{0.sign} (see
Proposition 4.5). Alternatively, this condition is always
satisfied if the singularities of the nonlinearity $f$
are strong enough, namely, if
$$
\lim_{u\to\pm1}F(u)=\infty,\ \ F(u):=\int_0^uf(s)\,ds,
$$
see Section 4. Furthermore, using some proper modification of the Moser
iteration scheme, we can also show that any trajectory $u(t)$
is separated from the singularities $\pm1$ if, in addition,
$$
\frac{f(u)}u\ge \frac {C}{(1-u^2)^p},\ \  p>1
$$
(see Remark 4.9; see also \cite{ERRCGM} for a similar condition for the Caginalp system). In that case, we  have $|u(t,x)|\le 1-\delta$ for some
$\delta>0$ and, consequently, the problem becomes factually
nonsingular and can be further investigated by using the techniques
devised for the Cahn-Hilliard equation with regular potentials.
Unfortunately, this last condition is not satisfied by the
physically relevant logarithmic potentials and we indeed need the
variational inequalities (and solutions) in order to deal with
such a potential.

The next, natural, step is to study the asymptotic behavior of the
system. In particular, we are interested here in the study of
finite-dimensional global attractors. We recall that the global
attractor is the smallest compact set of the phase space which is
fully invariant by the flow and attracts the bounded sets of initial
data as time goes to infinity; it thus appears as a suitable object in
view of the study of the longtime behavior of the
problem. Furthermore, when the global attractor has finite dimension
(in the sense of covering dimensions such as the fractal and the
Hausdorff dimensions), then, even though the initial phase space is
infinite-dimensional, the dynamics of the system is, in some proper
sense, finite-dimensional and can be described by a finite number of
parameters. We refer
the reader to, e.g., \cite{BV}, \cite{handbook}, \cite{temam} and the
references therein for
extensive reviews and discussions on this subject. One powerful
method, in order to prove the existence of the finite-dimensional
global attractor, is to prove the existence of a so-called exponential
attractor (in particular, this approach does not necessitate, contrary
to the usual one, based on the Lyapunov exponents, the
differentiability of the underlying semigroup). An exponential
attractor is a compact and semiinvariant set which contains the
global attractor, has finite fractal dimension and attracts all
bounded sets of initial data at an exponential rate. We refer the
reader to, e.g., \cite{EFNT}, \cite{EMZ1}, \cite{EMZ2} and
\cite{handbook} for more details and discussions on exponential attractors.

We thus prove the existence of global and exponential attractors for our
problem. We emphasize that this result is obtained under general assumptions (without any
sign assumption or any assumption of the form \eqref{0.reg}) and is thus
valid for the  variational solutions (which may not be solutions in
the usual sense). In particular, such solutions may reach the
singularities $\pm1$ on sets of positive measure on the
boundary $\R^+\times\Gamma$ or even on the whole boundary
$\R^+\times\Gamma$. This fact does not allow us to use the techniques
devised in \cite{MZ1} to establish the existence of finite-dimensional attractors
for the singular Cahn-Hilliard system with usual boundary
conditions (these techniques are strongly based on the
fact that $u(t)$ is separated from the singularities for almost all $t\ge0$, which is not
true in our case in general). Instead, we prove the
finite-dimensionality of the global attractor by using a proper
modification of the techniques developed in \cite{EZ1} for porous media equations.

This paper is organized as follows. In Section 2, we define proper
(regular) approximations of the singular potential and derive uniform
(with respect to these approximations) a priori estimates which allow
us, in Section 3, to formulate the variational inequality associated with the singular Cahn-Hilliard system with dynamic
boundary conditions and verify the existence and uniqueness of a solution for this inequality. We also  study
the further regularity of the solutions. Then, in Section 4, we give sufficient conditions which
ensure that the solutions are separated from the singularities of $f$
and, thus, satisfy the equations in the usual (distribution) sense. Section 5 is
devoted to the asymptotic behavior of the system. Finally, we give, in
Appendix 1, several auxiliary results. We also construct a simple example
which shows that the solutions may
not satisfy the dynamic boundary conditions in the usual sense for
logarithmic potentials.

\section{Approximations and uniform a priori estimates}\label{s1}

We consider the following equations (for simplicity, we set all constants equal to $1$):
\begin{equation}\label{1.main}
\begin{cases}
\Dt u=\Dx \mu,\ \ \partial_n\mu\big|_{\Gamma}=0,\\
\mu=-\Dx u+\tilde f(u)+h_1,\ \ u\big|_{t=0}=u_0,
\end{cases}
\end{equation}
in a bounded smooth domain $\Omega$ of $\R^3$, endowed with
dynamic boundary conditions on $\Gamma :=\partial \Omega $,
\begin{equation}\label{1.dyn}
\Dt\psi-\Delta_\Gamma \psi+g(\psi)+\partial_n u=h_2,\ \
\psi:=u\big|_{\Gamma}.
\end{equation}
Here, $u$ and $\mu$ are unknown functions, $\Dx$ and
$\Delta_\Gamma$ are the Laplace and Laplace-Beltrami operators on
$\Omega$ and $\Gamma $, respectively, $\tilde f$ and $g$
are known nonlinearities, $h_1\in L^2(\Omega)$ and $h_2\in
L^2(\Gamma)$
are given external forces and $\partial_n$ stands for the normal
derivative, $n$ being the unit outer normal to $\Gamma $.
\par
We assume that the nonlinearity $\tilde f$ has the form
\begin{equation}\label{1.str}
\tilde f(z):=f(z)-\lambda z,
\end{equation}
where $\lambda\in\R$ is a given constant and the singular function
$f$ satisfies
\begin{equation}\label{A.2}
\begin{cases}
1.\ \ f\in C^2((-1,1)),\\
2.\ \ f(0)=0,\ \ \lim_{u\to\pm1}f(u)=\pm\infty,\\
3.\ \ f'(u)\ge0,\ \ \lim_{u\to\pm1}f'(u)=+\infty,\\
4. \ \ \sgn u\cdot f''(u)\ge0.
\end{cases}
\end{equation}
Since the function $f$ is defined on the interval $(-1,1)$ only
and has singularities at $\pm1$, we a priori assume that
\begin{equation}\label{1.bound}
|u(t,x)|<1\ \text{almost everywhere in}\ \R^+\times\Omega.
\end{equation}
We finally assume that the second nonlinearity $g$ is
regular on the segment $[-1,1]$,
\begin{equation}\label{1.g}
g\in C^2([-1,1]).
\end{equation}
Then, we can assume, without loss of generality, that $g$ is
smoothly extended to the whole line, $g\in C^2(\R)$, and $g(z)=z+g_0(z)$ with $\|g_0\|_{C^2(\R)}\le C$ for some positive
constant $C$.
\par
In order to solve the singular problem \eqref{1.main}, we
approximate the nonlinearity $f$ by the following family of
smooth functions:
\begin{equation}\label{1.f}
f_N(u):=\begin{cases}
f(u),\ \ |u|\le 1-1/N,\\
f(1-1/N)+f'(1-1/N)(u-1+1/N), \ \ u>1-1/N,\\
f(-1+1/N)+f'(-1+1/N)(u+1-1/N),\ \ u<-1+1/N,
\end{cases}
\end{equation}
and we set $\tilde f_N(u):=f_N(u)-\lambda u$. We then consider the
approximate problems
\begin{equation}\label{1.regmain}
\begin{cases}
\Dt u=\Dx \mu,\ \ \partial_n\mu\big|_{\Gamma}=0,\\
\mu=-\Dx u+\tilde f_N(u)+h_1,\ \ u\big|_{t=0}=u_0,
\end{cases}
\end{equation}
endowed with the same dynamic boundary conditions \eqref{1.dyn}.
\par
The main aim of the present section is to derive several
uniform (with respect to $N\to\infty$) a priori estimates for the
solutions $(u,\mu)=(u_N,\mu_N)$ of problems \eqref{1.regmain},
\eqref{1.dyn} which will allow us (in the next section) to pass to
the limit $N\to\infty$ and establish the existence of a solution
for the singular problem (the existence, uniqueness and regularity
of solutions for the regular case, such as in the approximate problems
\eqref{1.regmain}, \eqref{1.dyn}, are now well-understood and will
not be considered in the present paper, see \cite{Gal}, \cite{GMS}, \cite{MZ2},
\cite{PRZ} and \cite{RZ} for
detailed expositions and related problems).
\par
As usual, it is convenient to rewrite problem \eqref{1.regmain}
in an equivalent form by using the inverse Laplacian
$A:=(-\Dx)^{-1}$ (endowed with Neumann boundary conditions). To
be more precise, since the first eigenvalue of the Laplacian with Neumann boundary
conditions vanishes, we assume that the operator $A$ is defined on
the functions with zero mean value only and maps them onto the functions with
zero mean value as well. Then, applying this operator to both sides of
\eqref{1.regmain}, we have
\begin{equation}\label{1.req}
A\Dt u:=(-\Dx)^{-1}\Dt u=\Dx u-\tilde f_N(u)-h_1+\<\mu\>,
\end{equation}
where $\<v\>$ stands for the mean value of the function $v$ over
$\Omega$. Furthermore, taking into account \eqref{1.dyn}, we see that
\begin{equation}\label{1.mu}
\<\mu\>=-\<\Dx u\>+\<\tilde
f_N(u)\>+\<h_1\>=\Dt\<u\>_\Gamma+\<g(u)\>_\Gamma-\<h_2\>_\Gamma-\<\tilde
f_N(u)\>+\<h_1\>,
\end{equation}
where $\<v\>_\Gamma:=\frac1{|\Omega|}\int_{\Gamma}v(x)dS$. We also
mention that problem \eqref{1.regmain} possesses the mass conservation
law
\begin{equation}\label{1.mass}
\<u(t)\>\equiv \<u(0)\>=c
\end{equation}
and, thus, $\<\Dt u\>=0$ and the left-hand side of
\eqref{1.req} is well-defined. Finally, keeping in mind the
singular limit $N\to\infty$, we only consider the initial
data $u_0$ for which $c\in(-1,1)$.
\par
We start with the usual energy equality.
\begin{lemma}\label{Lem1.en} Let the above assumptions hold and
let $u$ be a sufficiently regular
solution of \eqref{1.req}. Then, the following identity holds:
\begin{multline}\label{1.energy}
\frac d{dt}(\frac12\|\Nx u(t)\|^2_{L^2(\Omega)}+
\frac12\|\nabla_\Gamma u(t)\|^2_{L^2(\Gamma)}+(\tilde F_N(u(t)),1)_{\Omega}+(h_1,u(t))_{\Omega}+\\+
(G(u(t)),1)_{\Gamma}
-(h_2,u(t))_{\Gamma})+\|\Dt u(t)\|^2_{H^{-1}(\Omega)}+\|\Dt
u(t)\|^2_{L^2(\Gamma)}=0,
\end{multline}
where $\tilde F_N(z):=\int_0^z\tilde f_N(s)\,ds$,
$G(z):=\int_0^zg(s)\,ds$, $(\cdot,\cdot)_\Omega$ and
$(\cdot,\cdot)_{\Gamma}$ stand for the inner products in
$L^2(\Omega)$ and $L^2(\Gamma)$, respectively, and
$\|z\|_{H^{-1}(\Omega)}^2:=(Az,z)_\Omega$.
\end{lemma}
Indeed, multiplying \eqref{1.req} by $\Dt u$, integrating over $\Omega $ and by parts and
taking into account \eqref{1.dyn}, together with the identity
$\<\Dt u\>=0$, we deduce \eqref{1.energy}.

\begin{corollary}\label{Cor1.en} Let the above assumptions hold
and let, in addition, $N$ be large enough. Then, any (sufficiently regular) solution $u$
of problem \eqref{1.req} satisfies:
\begin{multline}\label{1.en-est}
\|u(t)\|^2_{H^1(\Omega)}+\|u(t)\|^2_{H^1(\Gamma)}+(F_N(u(t)),1)_\Omega+\\+\int_0^t(\|\Dt
u(s)\|^2_{H^{-1}(\Omega)}+\|\Dt u(s)\|^2_{L^2(\Gamma)})\,ds\le\\\le
C(\|u(0)\|^2_{H^1(\Omega)}+\|u(0)\|^2_{H^1(\Gamma)}+(F_N(u(0)),1)_\Omega+\|h_1\|^2_{L^2(\Omega)}+
\|h_2\|^2_{L^2(\Gamma)}),
\end{multline}
where $F_N(z):=\int_0^zf_N(s)\,ds$ and the constant $C$ is
independent of $t$ and $u(0)$.
\end{corollary}
Indeed, owing to our assumptions on $f$ and the explicit form of the
approximations $f_N$, see \eqref{1.f}, we can easily show that
\begin{equation}\label{1.flambda}
2F_N(z)+C\ge\tilde F_N(z)\ge \frac12 F_N(z)-C
\end{equation}
if $N\ge N_0(\lambda)$ is large enough, where the constant $C$ only
depends on $\lambda$. Integrating now \eqref{1.energy} with respect to
$t$ and using \eqref{1.flambda}, the fact that $g_0(u)$ is globally
bounded and obvious estimates, we end up with \eqref{1.en-est}.
\par
As a next step, we obtain the {\it dissipative} analogue of
estimate \eqref{1.en-est}.

\begin{lemma} \label{Lem1.sm} Let the assumptions of Lemma \ref{Lem1.en}
hold, $u$ be a sufficiently regular solution of \eqref{1.req} and
$N$ be large enough (depending on $\lambda$ and $c=\<u_0\>$).
Then, the following estimate holds:
\begin{multline}\label{1.dis}
\|u(t)\|^2_{H^1(\Omega)}+\|u(t)\|^2_{H^1(\Gamma)}+(F_N(u(t)),1)_\Omega+\\+\int_t^{t+1}(\|\Dt
u(s)\|^2_{H^{-1}(\Omega)}+\|\Dt u(s)\|^2_{L^2(\Gamma)}+\|f_N(u(s))\|_{L^1(\Omega)})\,ds\le\\\le
C(\|u(0)\|^2_{H^1(\Omega)}+\|u(0)\|^2_{H^1(\Gamma)}+(F_N(u(0)),1)_\Omega)e^{-\alpha t}+C(1+\|h_1\|^2_{L^2(\Omega)}+
\|h_2\|^2_{L^2(\Gamma)}),
\end{multline}
where the positive constants $C$ and $\alpha$ are independent of $N$ and
$u$, but can depend on the value $c$ in the mass conservation
 \eqref{1.mass}. In addition, the following smoothing
property holds:
\begin{multline}\label{1.smoothing}
\|u(t)\|^2_{H^1(\Omega)}+\|u(t)\|^2_{H^1(\Gamma)}+(F_N(u(t)),1)_\Omega\le\\\le
Ct^{-1}(\|u(0)-c\|_{H^{-1}(\Omega)}^2+\|u(0)\|^2_{L^2(\Gamma)}+\|h_1\|^2_{L^2(\Omega)}+\|h_2\|^2_{L^2(\Gamma)}+1),
\ \ t\in (0,1],
\end{multline}
where the constant $C$ is independent of $N$.
\end{lemma}
\begin{proof} We have, owing to assumptions \eqref{A.2} on the
nonlinearity $f$ and the fact that $c\in(-1,1)$,
\begin{equation}\label{1.imp}
\tilde f_N(z).(z-c)\ge \alpha f_N(z)z-C\ge\alpha/2|f_N(z)|-C_1,\ \ z\in\R,
\end{equation}
where $N$ is large enough and the positive constants $\alpha$ and $C_i$
depend on $c$ and $\lambda$, but are independent of $N$ (see \cite{MZ2}).
Multiplying now equation \eqref{1.req} by $\bar u(t):=u(t)-c$ and
using the above inequality, we find
\begin{multline}\label{1.est}
\frac12 \frac d{dt}(\|\bar u(t)\|^2_{H^{-1}(\Omega)}+\|\bar
u(t)\|^2_{L^2(\Omega)})+\alpha((f_N(u(t)),u(t))_\Omega+\\+\|\bar
u(t)\|^2_{H^1(\Omega)}+\|\bar u(t)\|^2_{H^1(\Gamma)})\le
C(1+\|h_1\|^2_{L^2(\Omega)}+\|h_2\|^2_{L^2(\Gamma)}),
\end{multline}
for some positive constants $\alpha$ and $C$. Applying the
Gronwall inequality to this relation, we obtain
\begin{multline}\label{1.grest}
\|\bar u(t)\|^2_{H^{-1}(\Omega)}+\|\bar
u(t)\|^2_{L^2(\Gamma)}+\\+\int_t^{t+1}(\|\bar
u(s)\|^2_{H^1(\Omega)}+\|\bar u(s)\|^2_{H^1(\Gamma)}+(f_N(u(s),u(s))_\Omega)\,ds
\le\\\le C(\|\bar u(0)\|^2_{H^{-1}(\Omega)}+\|\bar
u(0)\|^2_{L^2(\Gamma)})e^{-\alpha t}+C(1+\|h_1\|^2_{L^2(\Omega)}+\|h_2\|^2_{L^2(\Gamma)})
\end{multline}
for some positive constants $C$ and $\alpha$. In order to finish the proof
of the lemma, there only remains to note that, owing to the monotonicity of the function $f_N$,
\begin{equation}\label{1.pot}
F_N(z)\le f_N(z).z,\ \ z\in \R.
\end{equation}
Then, the smoothing property \eqref{1.smoothing} follows in a
standard way from \eqref{1.en-est}, \eqref{1.grest} and
\eqref{1.pot} and the dissipative estimate \eqref{1.dis} is an
immediate consequence of the dissipative estimate \eqref{1.grest}
(in a weaker norm) and the smoothing property \eqref{1.smoothing}, together with \eqref{1.en-est}.
This finishes the proof of Lemma \ref{Lem1.sm}.
\end{proof}
We are now ready to obtain additional regularity on $\Dt
u(t)$. To this end, we differentiate equation \eqref{1.req} with
respect to $t$
and set $\theta(t):=\Dt u(t)$. Then, this function solves
\begin{equation}\label{1.dif}
(-\Dx)^{-1}\Dt\theta=\Dx\theta-\tilde f_N'(u)\theta+\<\Dt\mu\>,\ \
\theta\big|_{t=0}=\theta_0,
\end{equation}
where $\theta_0:=-\Dx(\Dx u_0-\tilde f_N(u_0)-h_1)$, and
$$
\Dt\theta-\Delta_\Gamma\theta+\partial_n\theta+g'(u)\theta=0,\ \
\text{on $\Gamma$}.
$$

\begin{lemma}\label{Lem1.dt} Let the assumptions of Lemma
\ref{Lem1.en} hold. Then, the following estimate is valid for the derivative
$\theta(t):=\Dt u(t)$:
\begin{multline}\label{1.dtest}
\|\theta(t)\|_{H^{-1}(\Omega)}^2+\|\theta(t)\|^2_{L^2(\Gamma)}+\int_t^{t+1}(\|\theta(s)\|^2_{H^1(\Omega)}+
\|\theta(s)\|^2_{H^1(\Gamma)})\,ds\le\\\le
C(\|u(0)\|^2_{H^1(\Omega)}+\|u(0)\|^2_{H^1(\Gamma)}+\|\theta(0)\|^2_{H^{-1}(\Omega)}+
\|\theta(0)\|^2_{L^2(\Gamma)})e^{-\alpha
t}+\\+C(1+\|h_1\|^2_{L^2(\Omega)}+\|h_2\|^2_{L^2(\Gamma)}),
\end{multline}
where the positive constants $C$ and $\alpha$ can depend on the
total mass $c$, but are independent of $N$.
In addition, the following smoothing
property holds:
\begin{multline}\label{1.smdt}
\|\theta(t)\|^2_{H^{-1}(\Omega)}+\|\theta(t)\|_{L^2(\Gamma)}^2\le\\\le
Ct^{-2}(\|u(0)-c\|_{H^{-1}(\Omega)}^2+\|u(0)\|^2_{L^2(\Gamma)}+\|h_1\|^2_{L^2(\Omega)}+\|h_2\|^2_{L^2(\Gamma)}+1),
\ \ t\in (0,1],
\end{multline}
where the constant $C$ is independent of $N$.
\end{lemma}
\begin{proof} We multiply equation
\eqref{1.dif} by $\theta(t)$, integrate over $\Omega $ and use the fact that $\tilde
f'_N(u)\ge-\lambda$. Then, using also the boundary conditions and
the fact that $g'(u)$ is uniformly bounded, we find
\begin{multline}\label{1.est2}
\frac d{dt}(\|\theta(t)\|^2_{H^{-1}(\Omega)}+\|\theta
(t)\|_{L^2(\Gamma)}^2)+\alpha(\|\theta(t)\|^2_{H^1(\Omega)}+
\|\theta(t)\|_{H^1(\Gamma)}^2)\le \\ \le
C(\|\Dt u(t)\|^2_{L^2(\Omega)}+\|\Dt u(t)\|^2_{L^2(\Gamma)}),
\end{multline}
for some positive constants $\alpha$ and $C$ which are independent of $N$.
Interpolating between $H^{-1}$ and $H^1$ and applying the Gronwall inequality to this relation, we obtain the
desired estimate \eqref{1.dtest}. Combining this estimate with
\eqref{1.dis} and \eqref{1.smoothing} and arguing in a standard
way, we end up with \eqref{1.smdt} and finish the proof of the
lemma.
\end{proof}
The next lemma gives $H^1$-estimates on the solutions for every
fixed time $t\ge0$.
\begin{lemma}\label{Lem1,add1} Let the above assumptions hold.
Then, for every fixed $t\ge0$, the following estimate holds:
\begin{multline}\label{1.h1}
\|u(t)\|_{H^1(\Omega)}^2+\|u(t)\|_{H^1(\Gamma)}^2+\|f_N(u(t))\|_{L^1(\Omega)}\le\\\le
C(1+\|\Dt u(t)\|_{H^{-1}(\Omega)}^2+\|\Dt
u(t)\|_{L^2(\Gamma)}^2+\|h_1\|_{L^2(\Omega)}^2+\|h_2\|_{L^2(\Gamma)}^2),
\end{multline}
where the constant $C$ depends on $c$, but is independent of $t$
and $N$.
\end{lemma}
Indeed, multiplying equation \eqref{1.req} by $\bar u(t):=u(t)-c$
and arguing as in the derivation of \eqref{1.est} (but now
without integrating with respect to $t$), we deduce the desired estimate
\eqref{1.h1}. Here, we have used the inequality \eqref{1.imp} again.
\par
Furthermore, using \eqref{1.h1} and expression \eqref{1.mu} for
the mean value of $\mu$, we have
\begin{equation}\label{1.mumu}
|\<\mu(t)\>|\le C(1+\|\Dt u(t)\|_{H^{-1}(\Omega)}^2+\|\Dt
u(t)\|_{L^2(\Gamma)}^2+\|h_1\|_{L^2(\Omega)}^2+\|h_2\|_{L^2(\Gamma)}^2).
\end{equation}

We finally  rewrite
equation \eqref{1.req} in the form of a nonlinear elliptic
problem,
\begin{equation}\label{1.ell}
\begin{cases}
\Dx u(t)-f_N(u(t))-u(t)=\tilde h_1(t):=h_1-\\\text{\phantom{ooooooo0000000}}-u(t)-\lambda
u(t)+(-\Dx)^{-1}\Dt u(t)-\<\mu(t)\> ,\ \ \text{in $\Omega$},\\
\Delta_\Gamma u(t)-u(t)-\partial_n u(t)=\tilde h_2(t):=-h_2+g_0(u(t))+\Dt
u(t), \ \text{ on  $\Gamma$},
\end{cases}
\end{equation}
for every fixed $t$ and note that the estimates derived above
yield the following control of the right-hand sides in
\eqref{1.ell}:
\begin{multline}\label{1.good}
\|\tilde h_1(t)\|_{L^2(\Omega)}+\|\tilde
h_2(t)\|_{L^2(\Gamma)}\le\\\le
C(1+\|\Dt u(t)\|_{H^{-1}(\Omega)}^2+\|\Dt
u(t)\|_{L^2(\Gamma)}^2+\|h_1\|_{L^2(\Omega)}^2+\|h_2\|_{L^2(\Gamma)}^2)
\end{multline}
for some positive constant $C$ which is independent of $N$.
\par
Therefore, additional smoothness on the
solution $u:=u_N$ can be obtained by a proper elliptic regularity theorem (see \cite{MZ2}).
Unfortunately, in contrast to the case of regular potentials,
this problem does not satisfy the maximal regularity estimate in $L^2$ for
singular potentials $f$, see Appendix 1. Nevertheless, the
partial regularity formulated below is crucial for what follows.

\begin{lemma}\label{Cor1.H} Let the above assumptions hold and
set $\Omega_\eb:=\{x\in\Omega,\ \ d(x,\Gamma)>\eb\}$. Denote
by $n=n(x)$ some  smooth extension of the unit normal vector
field at the boundary inside the domain $\Omega$. Let also
$D_\tau u:=\Nx u-(\partial_n u)n$ be the
tangential part of the gradient $\Nx u$. Then, for every $\eb>0$,  the following
estimate holds:
\begin{multline}\label{1.max}
\|u(t)\|_{C^\alpha(\Omega)}+\|\Nx D_\tau
u(t)\|_{L^2(\Omega)}+\|u(t)\|_{H^2(\Omega_\eb)}+\|u(t)\|_{H^2(\Gamma)}\le\\\le
C_\eb(\|\tilde h_1(t)\|_{L^2(\Omega)}+\|\tilde h_2(t)\|_{L^2(\Gamma)})
\end{multline}
for some positive constants $\alpha$ and $C_\eb$ which are independent of~$N$.
\end{lemma}
The proof of this estimate is based on some variant of the
nonlinear localization technique and is given in Appendix 1 (see Theorem \ref{ThA.1}).
\par
We summarize the a priori estimates obtained so far in the following
theorem which is the main result of this section.

\begin{theorem}\label{Th1.mainest} Let the above assumptions hold and let $u$ be a
sufficiently regular solution of problem \eqref{1.req} with a
sufficiently large $N$ (depending on the constant $\lambda$ and
the total mass $c\in(-1,1)$). Then, the following estimate is
valid for every $\eb>0$:
\begin{multline}\label{1.regular}
\|u(t)\|_{C^\alpha(\Omega)}^2+\|u(t)\|^2_{H^2(\Gamma)}+
\|u(t)\|_{H^2(\Omega_\eb)}^2+\|u(t)\|^2_{H^1(\Omega)}+\\+
\|\Dt u(t)\|^2_{H^{-1}(\Omega)}+\|\Dt
u(t)\|^2_{L^2(\Gamma)}+\\+
\|\Nx D_\tau
u(t)\|^2_{L^2(\Omega)}+\|f_N(u(t))\|_{L^1(\Omega)}+\int_t^{t+1}(\|\Dt
u(s)\|^2_{H^1(\Omega)}+\|\Dt u(s)\|^2_{H^1(\Gamma)})\,ds\le\\
\le C(1+\|u(0)\|^2_{H^1(\Omega)}+\|u(0)\|^2_{H^1(\Gamma)}+\|\Dt u(0)\|^2_{H^{-1}(\Omega)}+
\|\Dt u(0)\|_{L^2(\Gamma)}^2)^2e^{-\beta t}+\\+
C(1+\|h_1\|^2_{L^2(\Omega)}+\|h_2\|^2_{L^2(\Gamma)})^2,
\end{multline}
where the positive constants $\alpha$, $\beta$ and $C$ (which can
depend on $\eb$) are independent of $N\to\infty$. In addition, the
following smoothing property holds:
\begin{multline}\label{1.smo}
\|\Dt u(t)\|_{H^{-1}(\Omega)}+\|\Dt u(t)\|_{L^2(\Gamma)}\le
C t^{-1}(\|u(0)-c\|_{H^{-1}(\Omega)}+\|u(0)\|_{L^2(\Gamma)}
+\\+\|h_1\|_{L^2(\Omega)}+\|h_2\|_{L^2(\Gamma)}+1),\ \ t\in (0,1],
\end{multline}
where the constant $C$ is also uniform with respect to
$N\to\infty$.
\end{theorem}

\begin{remark} \label{Rem1.reg} We thus have a uniform
$H^2$-estimate on the solution $u$ {\it inside} the domain and
an $L^2$-estimate on the gradient of the tangential derivatives
$\Nx D_\tau u$. In contrast to this, we do not have a uniform
control of the second normal derivative $\partial^2_n u$ close to the
boundary. Nevertheless, since the $L^1$-norm of the nonlinearity $f_N$
is controlled, \eqref{1.ell}, together with the
control of the tangential derivatives, allow us to estimate the
$L^1$-norm of $\partial_n^2 u$. We thus have the control
\begin{equation}\label{1.l1}
\|u(t)\|_{W^{2,1}(\Omega)}\le C(1+\|\tilde
h_1\|_{L^2(\Omega)}+\|\tilde h_2\|_{L^2(\Gamma)}).
\end{equation}
This, in turn, gives a control of the $L^1$-norm of the normal
derivative $\partial_n u$ at the boundary (owing to a proper trace
theorem),
\begin{equation}\label{1.trace}
\|\partial_n u(t)\|_{L^1(\Gamma)}\le C\|u(t)\|_{W^{2,1}(\Omega)}.
\end{equation}
As we will see in the next section, estimates \eqref{1.l1} and
\eqref{1.trace} remain true for the limit (as $N\to\infty$)
solution $u$ of the singular problem as well and, consequently,
the trace of $\partial_n u(t)$ at the boundary is
well-defined. However, owing to the nonreflexivity of
$L^1$-spaces, this trace {\it may not coincide}
 with the limit of $\partial_n u_N(t)\big|_{\Gamma}$ computed on the
 boundary by using the dynamic boundary condition \eqref{1.dyn}. Therefore,
 the boundary condition \eqref{1.dyn} may be violated
 for the limit singular solution. As we will see below, this indeed
 happens, even in the 1D case with smooth data. We overcome this
 difficulty by using monotonicity arguments and a proper variational
 formulation of problem \eqref{1.req}, see Section \ref{s2}.
\end{remark}

We conclude this section by establishing the standard uniform Lipschitz
continuity of the solution $u$ of problem \eqref{1.req} with
respect to the initial data.

\begin{proposition} Let the above assumptions hold and let
$u_1(t)$ and $u_2(t)$ be two (sufficiently regular) solutions of problem
\eqref{1.req} such that
$$
\<u_1(0)\>=\<u_2(0)\>=c.
$$
Then, the following estimate holds:
\begin{multline}\label{1.lip}
\|u_1(t)-u_2(t)\|_{H^{-1}(\Omega)}+\|u_1(t)-u_2(t)\|_{L^2(\Gamma)}\le\\\le
C(\|u_1(0)-u_2(0)\|_{H^{-1}(\Omega)}+\|u_1(0)-u_2(0)\|_{L^2(\Gamma)})e^{Kt},
\end{multline}
where the constants $C$ and $K$ are independent of $t$, $N$,
$u_1$ and $u_2$.
\end{proposition}
\begin{proof} Let $v(t)=u_1(t)-u_2(t)$. Then, this
function solves
\begin{equation}\label{1.difeq}
\begin{cases}
(-\Dx)^{-1}\Dt v-\Dx v+[\tilde f_N(u_1)-\tilde f_N(u_2)]=\<\mu_1-\mu_2\>,\ \ \text{in
$\Omega$},\\
\Dt v-\Delta_\Gamma v+\partial_n v+[g(u_1)-g(u_2)]=0,\ \ \text{on
$\Gamma$}.
\end{cases}
\end{equation}
Taking the scalar product of the first equation with $v(t)$,
integrating by parts and using the facts that $\<v(t)\>=0$, $\tilde
f'_N(z)\ge -\lambda$ and the nonlinearity $g'$ is globally
bounded, we obtain
\begin{multline}\label{1.difdif}
\frac
d{dt}(\|v(t)\|^2_{H^{-1}(\Omega)}+\|v(t)\|_{L^2(\Gamma)}^2)+\\+\alpha(\|v(t)\|^2_{H^1(\Omega)}+\|v(t)\|^2_{H^1(\Gamma)})\le
|\lambda|\|v(t)\|^2_{L^2(\Omega)}+C\|v(t)\|^2_{L^2(\Gamma)}
\end{multline}
for some positive constants $\alpha$ and $C$ which are independent of $N$.
Interpolating between $H^{-1}$ and $H^1$ in order to
estimate the $L^2$-norm of $v$ in $\Omega$ and applying the Gronwall inequality, we
find \eqref{1.lip} and finish the proof of the proposition.
\end{proof}

\section{The singular problem: variational
formulation and well-posedness}\label{s2}
The aim of this section is to pass to the limit $N\to\infty$ in
\eqref{1.req} and prove the existence and uniqueness of
solutions of the limit singular problem \eqref{1.main}. As
already mentioned, this limit solution is not necessarily a usual
distribution solution of the equations and we need to define it
in a proper way. To this end, we first fix a constant $L>0$ such
that
\begin{equation}\label{2.l}
\|\Nx u\|^2_{L^2(\Omega)}-\lambda \|u\|^2_{L^2(\Omega)}+L\|u\|_{H^{-1}(\Omega)}^2\ge
1/2\|u\|_{H^1(\Omega)}^2
\end{equation}
for all $u\in H^1(\Omega)$ with $\<u\>=0$ and introduce the
quadratic form
\begin{equation}\label{2.b}
B(u,v):=(\Nx u,\Nx
v)_{\Omega}-\lambda(u,v)_{\Omega}+L((-\Dx)^{-1}\bar u,\bar
v)_{\Omega}+(\nabla_\Gamma u,\nabla_\Gamma v)_{\Gamma},
\end{equation}
where $u,v\in H^1(\Omega)\cap H^1(\Gamma):=\{u,\ u\in H^1(\Omega ),\
u|_\Gamma \in H^1(\Gamma )\}$ (and a similar definition holds for
similar spaces) and $\bar u:=u-\<u\>$,
$\bar v=v-\<\bar v\>$. Then, obviously,
\begin{equation}\label{2.positive}
B(u,\bar u)=B(\bar u,\bar u)\ge0
\end{equation}
for all $u\in H^1(\Omega)\cap H^1(\Gamma)$.

The limit problem \eqref{1.req}, corresponding
to $N=\infty$, reads, {\it formally},
\begin{equation}\label{2.reg}
\begin{cases}
(-\Dx)^{-1}\Dt u=\Dx u-f(u)+\lambda u+\<\mu\>-h_1,\ \text{ in $\Omega$},\\
\mu:=-\Dx u+f(u)-\lambda u+h_1,\ \ u\big|_{\Gamma }=\psi,\\
\Dt \psi-\Delta_\Gamma \psi+g(\psi)+\partial_n u=h_2,\ \ \text{on $\Gamma$},\\
u\big|_{t=0}=u_0,\ \ \psi\big|_{t=0}=\psi_0.
\end{cases}
\end{equation}
We test the first equation (again formally) with the function $u-v$, where
$v=v(t,x)$ is smooth and satisfies
$$
\<u(t)-v(t)\>\equiv0.
$$
Then, after an integration by parts, we have
\begin{multline*}
(A\Dt u,u-v)_\Omega+(\Dt
u,u-v)_\Gamma+B(u,u-v)+(f(u),u-v)_{\Omega}=\\=L(Au,u-v)_\Omega-(g(u),u-v)_\Gamma-(h_1,u-v)_\Omega+(h_2,u-v)_\Gamma.
\end{multline*}
Finally, since $B$ is positive and $f$ is monotone, we have
$$
B(u,u-v)\ge B(v,u-v),\ \ \ (f(u),u-v)_\Omega\ge(f(v),u-v)_{\Omega},
$$
which yields
\begin{multline*}
(A\Dt u,u-v)_\Omega+(\Dt
u,u-v)_\Gamma+B(v,u-v)+(f(v),u-v)_{\Omega}\le\\\le L(Au,u-v)_\Omega-(g(u),u-v)_\Gamma-(h_1,u-v)_\Omega+(h_2,u-v)_\Gamma.
\end{multline*}
We recall that this inequality holds (again formally) for any
properly chosen test function $v$ such that $\<v(\tau)\>\equiv c$.
We are now ready to define a variational
solution of the limit problem \eqref{2.reg}.

\begin{definition}\label{Def2.var} Let
\begin{equation}\label{2.phsp}
(u_0,\psi_0)\in\Phi:=\{(u,\psi)\in L^\infty(\Omega)\cap
L^\infty(\Gamma),\ \ \|u\|_{L^\infty(\Omega)}\le 1,\ \
\|\psi\|_{L^\infty(\Gamma)}\le 1\}.
\end{equation}
A pair of functions $(u,\psi)$, $u=u(t,x)$, $x\in\Omega$,
$\psi=\psi(t,x)$, $x\in\Gamma$,
is a
variational solution of problem \eqref{2.reg} if
\begin{equation}\label{2.bb}
u(t)\big|_{\Gamma}=\psi(t)\ \ \text{for almost all $t>0$},\
u(0)=u_0,\ \ \psi(0)=\psi_0,
\end{equation}

\begin{equation}\label{2.bound}
\begin{aligned}
&1)\ \ -1<u(t,x)<1\ \ \text{for almost all
$(t,x)\in\R^+\times\Omega$},\\
&2)\ \ (u,\psi)\in C([0,\infty),H^{-1}(\Omega)\times L^2(\Gamma))\cap
L^2([0,T],H^1(\Omega)\times
H^1(\Gamma)),\  \forall T>0,\\
&3)\ \ f(u)\in L^1([0,T]\times\Omega),\ (\Dt u,\Dt\psi)\in
L^2([\tau,T], H^{-1}(\Omega)\times L^2(\Gamma)), \  \forall
T>\tau >0,
\end{aligned}
\end{equation}
$\<u(t)\>\equiv \<u(0)\>$ and the variational inequality
\begin{multline}\label{2.varmain}
(A\Dt u(t),u(t)-w)_\Omega+(\Dt
u(t),u(t)-w)_\Gamma+B(v,u(t)-w)+(f(w),u(t)-w)_{\Omega}\le\\\le L(Au(t),u(t)-w)_\Omega-
(g(u(t)),u(t)-w)_\Gamma-(h_1,u(t)-w)_\Omega+(h_2,u(t)-w)_\Gamma
\end{multline}
is satisfied for almost every $t>0$ and every test function $w=w(x)$ such that
$$
w\in H^1(\Omega)\cap H^1(\Gamma),\ \ f(w)\in L^1(\Omega)
$$
and $\<w\>=\<u(0)\>$.  Note that the
relation $u(t)\big|_{\Gamma}=\psi(t)$ is assumed to hold {\it
only} for $t>0$. At the initial time $t=0$, no relation between
$u_0$ and $\psi_0$ is assumed. However, for $t>0$ the function
$\psi$ can be found if $u$ is known. Therefore, we can indeed write
the variational inequality \eqref{2.varmain} in terms of the
function $u$ only.
\end{definition}
Before studying the existence and uniqueness of variational solutions, it is convenient
to rewrite the variational inequality in terms of test functions $v=v(t,x)$ depending on $t$ and $x$.
More precisely, let the test function $v$ satisfy the regularity assumptions \eqref{2.bound} and $\<v(t)\>\equiv\<u_0\>=c$
(we will call this class of functions {\it admissible test functions} below).
Then, we can write inequality \eqref{2.varmain} with $w=v(t)$ for almost all $t>0$. Moreover, due to the regularity
assumptions \eqref{2.bound} on $u$ and $v$, we see that all terms obtained are in $L^1$  with respect to $t$. Thus, we can
integrate this inequality with respect to $t$, which gives 
\begin{multline}\label{2.varmain1}
\int_s^t[(A\Dt u,u-v)_\Omega+(\Dt u,u-v)_\Gamma]d\,\tau
+\\+\int_s^t[B(v,u-v)+(f(v),u-v)_{\Omega}]\,d\tau \le\\\le\int_s^t [L(u,A(u-v))_\Omega-(g(u),u-v)_\Gamma-(h_1,u-v)_\Omega+(h_2,u-v)_\Gamma]\,d\tau 
\end{multline}
for all $t>s>0$.

The next theorem gives the uniqueness of such variational
solutions.
\begin{theorem}\label{Th2.unique} Let the nonlinearities $f$ and $g$
and the external forces $h_1$ and $h_2$ satisfy the assumptions of
Section \ref{s1}. Then, the variational solution of problem
\eqref{2.reg} (in the sense of Definition \ref{Def2.var}) is
unique and is independent of the choice of $L$ satisfying
\eqref{2.l}. Furthermore, for every two variational solutions $u_1$
and $u_2$ such that $\<u_1(0)\>=\<u_2(0)\>$, the following estimate holds:
\begin{multline}\label{2.lip}
\|u_1(t)-u_2(t)\|_{H^{-1}(\Omega)}+\|\psi_1(t)-\psi_2(t)\|_{
L^2(\Gamma)}\le\\\le
Ce^{Kt}(\|u_1(0)-u_2(0)\|_{H^{-1}(\Omega)}+\|\psi_1(0)-\psi_2(0)\|_{
L^2(\Gamma)}),
\end{multline}
where the constants $C$ and $K$ are independent of $t$, $u_1$ and
$u_2$.
\end{theorem}
\begin{proof} We first need to deduce one more variational
inequality for a solution $u$. Let $w$ be a test
function satisfying the assumptions of Definition \ref{Def2.var}
and set
$$
v_\alpha:=(1-\alpha)u+\alpha w,\ \ \alpha\in[0,1].
$$
Then, owing to assumption \eqref{A.2}(4), the function $|f(u)|$ is
convex and, therefore,
$$
|f(v_\alpha)|\le |f(u)|+|f(w)|.
$$
Consequently, $v_\alpha$ is an admissible test function for
every $\alpha\in[0,1]$. Inserting $v=v_\alpha$ in the variational
inequality \eqref{2.varmain1}, dividing it by $\alpha$ and using
the fact that $u\in AC([s,t], H^{-1}(\Omega)\cap
L^2(\Omega))$ (here, $AC$ stands for absolutely continuous), we see that
\begin{multline}\label{2.var1}
\int_s^t[(A\Dt u,u-w)_\Omega+(\Dt u,u-w)_{\Gamma}]\,d\tau +\\+\int_s^t[B(v_\alpha,u-w)+(f(v_\alpha),u-w)_{\Omega}]\,d\tau \le\\
\le\int_s^t [L(u,A(u-w))_\Omega-(g(u),u-w)_\Gamma-(h_1,u-w)_\Omega+(h_2,u-w)_\Gamma]\,d\tau .
\end{multline}
Passing to the limit $\alpha\to0$ in \eqref{2.var1} and using the
Lebesgue dominated convergence theorem for the nonlinear term,
we end up with the desired additional vartiational inequality, namely,
\begin{multline}\label{2.varadd}
\int_s^t[(A\Dt u,u-w)_\Omega+(\Dt u,u-w)_{\Gamma}]\,d\tau
+\\+\int_s^t[B(u,u-w)+(f(u),u-w)_{\Omega}]\,d\tau \le\\\le\int_s^t [L(u,A(u-w))_\Omega-(g(u),u-w)_\Gamma-(h_1,u-w)_\Omega+(h_2,u-w)_\Gamma]\,d\tau ,
\end{multline}
where $w=w(t,x)$ is an arbitrary admissible test function.
\par
We are now ready to prove the uniqueness. Let $u_1$ and
$u_2$ be two variational solutions of problem \eqref{2.reg}. We
consider the variational inequality \eqref{2.varmain1} with
$u=u_1$ and $v=u_2$, together with the additional variational
inequality \eqref{2.varadd} with $u=u_2$ and $w=u_1$ (this makes
sense, since $u_1$ and $u_2$ are admissible test functions), and
sum the two resulting inequalities. Then, the terms containing $B$,
$f$, $h_1$ and $h_2$ vanish and, using, in addition, the fact that
$u_i\in AC([s,t], H^{-1}(\Omega)\cap L^2(\Gamma))$, $i=1,2$, we end up with
the following inequality:
\begin{multline}\label{2.gr}
\frac12(\|u_1(t)-u_2(t)\|^2_{H^{-1}(\Omega)\cap
L^2(\Gamma)}-\|u_1(s)-u_2(s)\|^2_{H^{-1}(\Omega)\cap
L^2(\Gamma)})\le\\\le
\int_s^t[L\|u_1(\tau)-u_2(\tau)\|^2_{H^{-1}(\Omega)}-(g(u_1(\tau))-g(u_2(\tau)),u_1(\tau)-u_2(\tau))_{\Gamma}]\,d\tau,
\end{multline}
where $\|u\|^2_{H^{-1}(\Omega )\cap L^2(\Gamma )}
=\|u\|^2_{H^{-1}(\Omega )}+\|u\|^2_{L^2(\Gamma )}$. Using now the fact that $g\in C^1([-1,1])$ and applying the
Gronwall inequality to \eqref{2.gr}, we see that
$$
\|u_1(t)-u_2(t)\|^2_{H^{-1}(\Omega)\cap L^2(\Gamma)}\le
Ce^{K(t-s)}\|u_1(s)-u_2(s)\|^2_{H^{-1}(\Omega)\cap L^2(\Gamma)}
$$
for some positive constants $C$ and $K$ which are independent of
$t>s>0$ and $u_i$, $i=1,2$. Passing to the limit $s\to0$ in this estimate
and using the continuity \eqref{2.bound}(2) of $u_1$ and
$u_2$, we deduce the desired estimate \eqref{2.lip} which, in
particular, gives the uniqueness.
\par
Thus, we only need to prove that the above definition of a
solution is independent of the choice of $L$. To this end, we
assume that $u_1$ is a variational solution for $L=L_1$ and $u_2$
is a variational solution for
$L=L_2$. Let also $u_1(0)=u_2(0)$. Using then the obvious relation
\begin{multline*}
B_{L_1}(v,u_1-v)-L_1(u_1,A(u_1-v))_\Omega=\\=B_{L_2}(v,u_1-v)-L_2(u_1,A(u_1-v))_{\Omega}-(L_1-L_2)\|u_1-v\|_{H^{-1}(\Omega)}^2
\end{multline*}
and arguing exactly as in the proof of \eqref{2.lip}, we have
\begin{multline}\label{2.gr1}
\frac12(\|u_1(t)-u_2(t)\|^2_{H^{-1}(\Omega)\cap
L^2(\Gamma)}-\|u_1(s)-u_2(s)\|^2_{H^{-1}(\Omega)\cap
L^2(\Gamma)})\le\\\le
\int_s^t[L_1\|u_1(\tau)-u_2(\tau)\|^2_{H^{-1}(\Omega)}-(g(u_1(\tau))-g(u_2(\tau)),u_1(\tau)-u_2(\tau))_{\Gamma}]\,d\tau,
\end{multline}
which coincides with \eqref{2.gr} and, therefore, also leads to estimate
\eqref{2.lip}. Thus, $u_1\equiv u_2$ and Theorem \ref{Th2.unique}
is proved.
\end{proof}

We are now able to prove the existence of a variational solution
$u$ of problem \eqref{2.reg} by passing to the limit $N\to\infty$ in equations
\eqref{1.req}.

\begin{theorem}\label{Th2.main} Let the  assumptions of the previous theorem hold.
Then, for every pair $(u_0,\psi_0)\in\Phi$,
problem \eqref{2.reg} possesses a unique variational
solution $(u,\psi )$ in the sense of Definition \ref{Def2.var}. Furthermore,
this solution regularizes as $t>0$ and all the uniform estimates obtained in Section \ref{s1} hold
for the solutions of the the limit singular equation \eqref{2.reg}. In particular, the following
estimate is valid for every $\eb>0$:
\begin{multline}\label{2.regular}
\|u(t)\|_{C^\alpha(\Omega)}^2+\|u(t)\|^2_{H^2(\Gamma)}+\\+
\|u(t)\|_{H^2(\Omega_\eb)}^2+\|u(t)\|^2_{H^1(\Omega)}+
\|\Dt u(t)\|^2_{H^{-1}(\Omega)}+\|\Dt
u(t)\|^2_{L^2(\Gamma)}+\\+
\|\Nx D_\tau
u(t)\|^2_{L^2(\Omega)}+\|f(u(t))\|_{L^1(\Omega)}+\int_t^{t+1}(\|\Dt
u(s)\|^2_{H^1(\Omega)}+\|\Dt u(s)\|^2_{H^1(\Gamma)})\,ds\le\\
\le C\frac{t^4+1}{t^4}(1+\|h_1\|^2_{L^2(\Omega)}+\|h_2\|^2_{L^2(\Gamma)})^2,\ \ t>0,
\end{multline}
for some positive constants $\alpha$ and $C$ which are independent
of $t$ and $u$ (we recall that $\Omega_\eb:=\{x\in\Omega,\
d(x,\Gamma)>\eb\}$), where $D_\tau u$ denotes the tangential part of
$\nabla _xu$ (see Lemma 2.6); in addition, all
norms in the
left-hand side of \eqref{2.regular} make sense for any variational
solution $u$.
\end{theorem}
\begin{proof} Let $u_N$ be the solution of the
approximate problem \eqref{1.req}. Then, repeating the
derivation of the variational inequality \eqref{2.varmain1}, we see
that
\begin{multline}\label{2.varmainN}
\int_s^t[(A\Dt u_N,u_N-v)_{\Omega}+(\Dt u_N,u_N-v)_{\Gamma}]d\,\tau
+\\+\int_s^t(B(v,u_N-v)+(f_N(v),u_N-v)_{\Omega})\,d\tau \le\\\le\int_s^t [L(u_N,A(u_N-v))_\Omega-(g(u_N),u_N-v)_\Gamma-(h_1,u_N-v)_\Omega+(h_2,u_N-v)_\Gamma]\,d\tau
\end{multline}
for every admissible test function $v$ and every $t>s>0$ (we
recall that the solution $u_N$ of the regularized problem \eqref{1.req}
is smooth and all the formal calculations performed in the derivation
of \eqref{2.varmain} can be easily justified in that case).
\par
Our aim is to pass to the limit $N\to\infty$ in \eqref{2.varmainN}.
We start with the case where the initial datum $u_0$ is smooth and
satisfies the additional conditions
\begin{equation}\label{2.sep}
|u_0(x)|\le 1-\delta,\ \ \delta>0,\ \ \psi_0:=u_0\big|_\Gamma.
\end{equation}
Then, according to Theorem \ref{Th1.mainest}, the sequence $u_N$
satisfies the uniform estimate \eqref{1.regular} and, therefore, we can assume, without loss of generality, that $u_N$
converges to some limit function $u$ in the following sense:
\begin{equation}\label{2.conv}
\begin{aligned}
&1)\ \ \  u_N\to u\ \text{ weakly-$*$ in }
L^\infty([0,T],(H^1(\Omega)\cap H^2(\Gamma))\cap H^2(\Omega_\eb)),\\
&2)\ \ \ \Dt u_N\to\Dt u\ \text{ weakly-$*$ in }\
L^\infty([0,T],H^{-1}(\Omega)\cap L^2(\Gamma))\\
&\text{\phantom{rooooryyyyyyyiii} and weakly in }
L^2([0,T],H^1(\Omega)\cap H^1(\Gamma)),\\
&3)\ \ \  D^2_\tau u_N\to D^2_\tau u\ \text{ weakly-$*$ in }
L^\infty([0,T],L^2(\Omega)),\\
&4)\ \  \ u_N\to u\ \text{ strongly  in } C^\gamma([0,T]\times\Omega) \text{ for some $\gamma>0$}.
\end{aligned}
\end{equation}
Indeed, the last strong convergence follows from the facts that
$u_N$ is uniformly bounded in
$L^\infty([0,T],C^{\alpha}(\Omega))$, $\alpha >0$, and $\Dt u_N$ is uniformly
bounded in $L^\infty([0,T],H^{-1}(\Omega))$ (owing to Theorem
\ref{Th1.mainest} and the assumption that the initial datum $u_0$ is
smooth and is separated from the singularities $\pm1$).
\par
These convergence results allow us to pass to the limit $N\to\infty$  in \eqref{2.varmainN}
and prove that the limit function satisfies \eqref{2.varmain1} for
any admissible function $v$. The only nontrivial term containing
the nonlinearity $f_N$ can be treated by using the inequality
$$
|f_N(v)|\le |f(v)|,
$$
the fact that $f(v)\in L^1([0,T]\times\Omega)$ and the Lebesgue
dominated convergence theorem.
\par
Thus, we only need to show that the function $u$ thus
constructed satisfies the regularity assumptions \eqref{2.bound}.
The only nontrivial statements that we need to prove are
that \eqref{2.bound}(1) holds and $f(u)\in
L^1([0,T]\times\Omega)$ (the other ones are immediate
consequences of \eqref{2.conv}). Let us check the
first one. Since the $L^1$-norm of $f_N(u_N)$ is
uniformly bounded, we conclude from the expression of the function
$f_N$ that
\begin{equation}\label{2.mn}
{\rm meas}\{(t,x)\in[T,T+1]\times\Omega,\ |u_M(t,x)|>1-\frac1N\}\le
\varphi(\frac1N),\ \ M\ge N,
\end{equation}
where
\begin{equation}
\varphi(x):=\frac C{\max\{|f(1-x)|,|f(x-1)|\}}
\end{equation}
for some constant $C$ which is independent of $T\in\R^+$, $M\ge N$ and
$N\in\Bbb N$. Thus, passing to the limit $M,N\to\infty$ in
\eqref{2.mn} and using the fact that $\varphi(x)\to0$ as $x\to0$, we
conclude that
$$
{\rm meas}\{(t,x)\in[T,T+1]\times\Omega,\ |u(t,x)|=1\}=0
$$
and \eqref{2.bound}(1) is verified. In order to prove that $f(u)$
is integrable, there only remains to note that the already proved statement
\eqref{2.bound}(1), together with the convergence \eqref{2.conv}(4),
imply the almost everywhere convergence $f_N(u_N)\to f(u)$ and,
therefore, owing to the Fatou lemma,
\begin{equation}\label{2.f}
\|f(u)\|_{L^1([T,T+1]\times\Omega)}\le
\operatorname{liminf}_{N\to\infty}\|f_N(u_N)\|_{L^1([T,T+1]\times\Omega)}<\infty.
\end{equation}
Thus, $u$ is indeed the desired variational solution of
\eqref{2.reg} and estimate \eqref{2.regular} immediately follows
from \eqref{2.conv} and Theorem \ref{Th1.mainest} (in order to deduce the $L^1$-estimate
on $f(u)$, one needs to use, in addition, inequality
\eqref{2.f}).
\par
Finally, we are now able to remove assumption \eqref{2.sep}. To
this end, we approximate the initial datum $(u_0,\psi_0)\in\Phi$ by
a sequence $(u_0^k,\psi_0^k)$ of smooth functions satisfying \eqref{2.sep} (of
course, with $\delta=\delta_k$ which can tend to zero as
$k\to\infty$) in such a way that
\begin{equation}\label{2.appin}
\|u_0-u_0^k\|_{L^2(\Omega)}\to0,\ \
\|u_0^k\big|_\Gamma-\psi_0\|_{L^2(\Gamma)}\to0,\ \
\<u_0^k\>\equiv\<u_0\>.
\end{equation}
Let $(u_k(t),\psi_k(t))$ (where $\psi_k=u_k\big|_\Gamma$) be a
sequence of variational solutions of problem \eqref{2.reg}
satisfying $(u_k(0),\psi_k(0))=(u_0^k,\psi_0^k)$ (whose existence
is proved above). Then, owing to the uniform Lipschitz continuity
estimate \eqref{2.lip} and assumption \eqref{2.appin},
$(u_k,\psi_k)$ is a Cauchy sequence in
$C([0,T],H^{-1}(\Omega)\times L^2(\Gamma))$ and, therefore, the
limit function
$$
(u,\psi):=\lim_{k\to\infty} (u_k,\psi_k)
$$
exists and also belongs to $C([0,T],H^{-1}(\Gamma)\times
L^2(\Omega))$. The fact that $u$ is a variational solution of
\eqref{2.reg}, as well as estimate \eqref{2.regular}, can be verified, based
on the uniform estimates derived in Section \ref{s1}, exactly as was
done above for smooth initial data. This finishes the proof of Theorem
\ref{Th2.main}.
\end{proof}

\begin{corollary}\label{Cor2.sem} Under the assumptions of Theorem
\ref{Th2.main}, equation \eqref{2.reg} generates a solution
semigroup $S(t)$ in the phase space $\Phi$,
\begin{equation}\label{2.sem}
S(t)(u_0,\psi_0):=(u(t),\psi(t)),\ \ S(t):\Phi\to\Phi,\ \ t\ge0,
\end{equation}
where $(u(t),\psi(t))$ is the unique variational solution of problem
\eqref{2.reg} with initial datum $(u_0,\psi_0)$. Furthermore, this
semigroup is globally Lipschitz continuous,
\begin{equation}\label{2.slip}
\|S(t)(u_0^1,\psi_0^1)-S(t)(u_0^2,\psi_0^2)\|_{H^{-1}(\Omega)\times
L^2(\Gamma)}\le
Ce^{Kt}\|(u^1_0-u_0^2,\psi_0^1-\psi_0^2)\|_{H^{-1}(\Omega)\times
L^2(\Gamma)},
\end{equation}
in the metric of the space $\Phi^w:=H^{-1}(\Omega)\times
L^2(\Omega)$.
\end{corollary}

We now start investigating the analytic structure of a
 solution $(u(t),\psi(t))$ of problem \eqref{2.reg}
(this will be continued in Section \ref{s3}).
\begin{proposition}\label{Prop2.int} Let $(u(t),\psi(t))$ be a variational
solution of problem \eqref{2.reg} constructed in Theorem
\ref{Th2.main}. Then, $\psi(t)=u(t)\big|_{\Gamma}$ for $t>0$ and, for every $\varphi\in
C_0^\infty((0,T)\times\Omega)$ such that $\<\varphi(t)\>\equiv0$,
there holds
\begin{multline}\label{2.dist}
\int_{\R^+}((-\Dx)^{-1}\Dt u(t),\varphi(t))_\Omega\,dt=\\=\int_{\R^+}((\Dx
u(t),\varphi(t))_\Omega-(f(u(t)),\varphi(t))_{\Omega}+\lambda(u(t),\varphi(t))_{\Omega}-(h_1,\varphi(t))_\Omega)\,dt.
\end{multline}
Furthermore,

\begin{equation}\label{2.r}
u\in L^\infty([\tau,T],W^{2,1}(\Omega)),\ \ T>\tau>0,
\end{equation}

\noindent and the trace of
the normal derivative on the boundary,
\begin{equation}\label{2.trace}
[\partial_n u]_{int}:=\partial_n u\big|_{\Gamma}\in
L^\infty([\tau ,T],L^1(\Gamma)),\ \ T>\tau >0,
\end{equation}
exists.
\end{proposition}
\begin{proof} Since, according to Theorem
\ref{Th1.mainest}, the approximating sequence $u_N$ is uniformly
bounded in $L^\infty([\tau,T],H^2(\Omega_\eb))$, for any $\eb>0$,
the sequence $f_N(u_N)$ is also uniformly bounded in
$L^\infty([\tau,T],L^2(\Omega_\eb))$. This fact, together with the
almost everywhere convergence established in the proof of Theorem
\ref{Th2.main}, guarantee that $f_N(u_N)\to f(u)$ weakly in
$L^2([\tau,T]\times\Omega_\eb)$ for all $\eb>0$ and this, in
turn, allows us to verify identity \eqref{2.dist} by passing to
the limit in the analogous identity for the approximate
solutions $u_N$.
\par
In order to check the remaining statements of the proposition, we
first deduce from \eqref{2.dist}  that
\begin{equation}\label{2.di}
(-\Dx)^{-1}\Dt u(t)=\Dx u(t)-f(u(t))+\lambda u(t)-h_1+c(t)
\end{equation}
for some function $c\in L^\infty([\tau,T])$, $T>\tau >0$. At this point,
equality \eqref{2.di} is understood as an equality in
$L^2_{\rm loc}([\tau,T]\times\Omega)$. However, owing to estimate
\eqref{2.regular}, $f(u)\in L^\infty([\tau,T], L^1(\Omega))$ and
the term $(-\Dx)^{-1}\Dt u$ also belongs at least to this space.
Thus, we see that
$$
\Dx u\in L^\infty([\tau,T],L^1(\Omega))
$$
and, consequently, since $\nabla _xD_\tau u$ is
controlled by \eqref{2.regular}, we can finally conclude
that \eqref{2.r} holds,
which gives the existence of the trace \eqref{2.trace} and
finishes the proof of Proposition \ref{Prop2.int}.
\end{proof}

Note that, using the obvious fact that $\<(-\Dx)^{-1}\Dt u\>=0$, we
can find the explicit formula for the function $c(t)$ in \eqref{2.di}, namely,
\begin{equation}\label{2.c}
c(t)=\<\Dx u(t)-f(u(t))+\lambda u(t)-h_1\>=-\<\mu(t)\>
\end{equation}
and, therefore, the first equation of \eqref{2.reg} is satisfied
in a usual sense (say, as an equality in
$L^2_{loc}([\tau,T]\times\Omega)$ or/and almost everywhere).
\par
We now investigate the third equation of \eqref{2.reg} (the
equation on the boundary). According to Theorem
\ref{Th1.mainest}, we see that the approximating sequence
$(u_N(t),\psi_N(t))$ satisfies
$$
\|\Dt\psi_N(t)\|_{L^\infty([\tau,T],L^2(\Gamma))}+
\|\psi_N(t)\|_{L^2([\tau,T],H^2(\Gamma))}\le C
$$
and, therefore, using the fact that the approximate solutions satisfy the
second equation of \eqref{2.reg}, we
can assume, without loss of generality, that we have the convergence
\begin{equation}\label{2.ext}
[\partial_n u]_{ext}:=\lim_{N\to\infty}\partial_n
u_N\big|_{\Gamma}\in L^\infty([\tau,T],L^2(\Gamma)),\ \ T>\tau>0,
\end{equation}
where the limit is understood as a weak-star limit in
$L^\infty([\tau,T],L^2(\Gamma))$. Then, obviously,
\begin{equation}\label{2.bou}
\Dt\psi-\Delta_\Gamma\psi+g(\psi)+[\partial_n u]_{ext}=h_2, \ \
\text{on $\Gamma$},
\end{equation}
and, in order to verify that the variational solution $(u,\psi)$
satisfies equations \eqref{2.reg} in the usual sense, there only
remains to check that
\begin{equation}\label{2.intext}
[\partial_n u]_{int}=[\partial_n u]_{ext}\ \ \text{for almost every
$(t,x)\in\R^+\times\Gamma$}.
\end{equation}
However, as the example in Appendix 1 shows, this identity can be
violated even in the simplest 1D stationary case. In the next section, we formulate several sufficient conditions which
ensure that
\eqref{2.intext} holds for every (variational) solution of
\eqref{2.reg}.

\section{Additional regularity and separation from the
singularities}\label{s3}
The main aim of this section is to study the analytic properties
of the variational solutions $u$ of problem \eqref{2.reg}, especially
close to the singular points $\pm1$. We start with the following
result which gives an additional regularity on
$u(t,x)$ close to the points where $|u(t,x)|<1$.

\begin{proposition}\label{Prop3.h2reg} Let the assumptions of
Theorem \ref{Th2.unique} hold and let $u$ be a variational
solution of problem \eqref{2.reg}. Let also $\delta>0$, $T>0$ be given
and set
\begin{equation}\label{3.odelta}
\Omega_\delta(T):=\{x\in\Omega,\ \ |u(T,x)|<1-\delta\}.
\end{equation}
Then, $u\in W^{2,2}(\Omega_\delta(T))$ and the
following estimate holds:
\begin{equation}\label{3.rest}
\|u\|_{H^{2}(\Omega_\delta(T))}\le Q_{\delta,T},
\end{equation}
where the constant $Q_{\delta,T}$ only depends on $T$ and
$\delta$, but is independent of the concrete choice of the
solution $u$.
\end{proposition}
\begin{proof}
Since the solution $u(T,x)$ is H\"older continuous with respect to $x$ (see \eqref{2.regular}),
 there exists a smooth nonnegative cut-off function $\theta(x)$ such
that
\begin{equation}\label{3.cut}
\begin{cases}
1) \ \ \theta(x)\equiv1,\ \ x\in \Omega_\delta(T),\\
2)\ \ \theta(x)\equiv0,\ \
x\in\Omega\backslash\Omega_{\delta/2}(T),\\
3)\ \ \|\theta\|_{C^2(\R^3)}\le K_{\delta,T},
\end{cases}
\end{equation}
where $K_{\delta,T}$ depends on the constants in
\eqref{2.regular}, but is independent of the concrete choice of
the solution $u$.
\par
Furthermore,
let $u_N(t,x)$ be a sequence of approximate solutions
of problems \eqref{1.req} which converges to the variational
solution $u(t,x)$ as $N\to\infty$. Then, since this convergence
holds in the space $C^\gamma([t,T]\times\Omega)$ for some
$\gamma>0$,
\begin{equation}\label{3.nbound}
|u_N(T,x)|<1-\delta/4,\ \ x\in\Omega_{\delta/2}(T)
\end{equation}
if $N$ is large enough. Set now $v_N(x):=\theta(x)u_N(T,x)$. Then,
this function obviously solves the following elliptic boundary
value problem (compare with \eqref{1.ell}):
\begin{equation}\label{3.smooth}
\begin{cases}
\Dx v_N-v_N=h_1(u_N):=\theta f_N(u_N(T))+\theta\tilde h_1(T)+2\Nx\theta .\Nx
u_N(T)+u_N(T)\Dx\theta ,\\
v_N\big|_{\Gamma}=w_N,\\
\Delta_\Gamma w_N-w_N-\partial_n v_N=h_2(u_N):=\\\text{\phantom{eggogeggog}}:=\theta\tilde
h_2(T)+2\nabla_\Gamma \theta .\nabla_\Gamma u_N(T)+u_N(T)\Delta_\Gamma\theta
-u_N(T)\partial_n\theta ,
\end{cases}
\end{equation}
where the functions $\tilde h_i$, $i=1,2$, are the same as in \eqref{1.ell}.
In addition, owing to estimates \eqref{1.smdt}, \eqref{1.h1},
\eqref{1.good}, \eqref{3.cut} and
\eqref{3.nbound}, we see that
\begin{equation}\label{3.l2}
\|h_1(u_N)\|_{L^2(\Omega)}+\|h_2(u_N)\|_{L^2(\Gamma)}\le Q_{\delta,T},
\end{equation}
where the constant $Q_{\delta,T}$ is independent of $N$ and of the
concrete choice of the solution $u$. Applying the $H^2$-regularity
theorem to the linear elliptic problem \eqref{3.smooth} (see \cite{MZ2}) and
recalling \eqref{3.cut}, we deduce that
\begin{equation}\label{3.c}
\|u_N(T)\|_{H^{2}(\Omega_\delta(T))}\le Q_{\delta,T}
\end{equation}
and, consequently, by passing to the limit $N\to\infty$, we see
that $u(T)\in H^{2}(\Omega_\delta(T))$ and \eqref{3.rest}
holds. This finishes the proof of the proposition.
\end{proof}

\begin{remark}\label{Rem3.reg} Applying the $L^p$-regularity
theorem to the elliptic boundary value problem \eqref{3.smooth},
together with a proper interpolation inequality, we have
$$
\|u(T)\|_{W^{2,p}(\Omega_\delta(T))}\le Q_{\delta,p,T}
$$
for any $p<\infty$. However, it seems difficult to obtain further regularity results on
$u$ in $\Omega_\delta(T)$ by directly using equations \eqref{1.req} or
\eqref{3.smooth}, owing to the presence of the nonlocal term
$\<\mu(t)\>$ (which is only $L^\infty$ with respect to $t$).
Alternatively, one can use the standard interior estimates for the
initial fourth-order problem \eqref{1.main}. Then,
it is not difficult to see that the factual regularity of the
solution $u$ in $\Omega_\delta(T)$ is only restricted by the
regularity of the data $f$, $g$, $h_i$, $i=1,2$, and $\Omega$ (and, if these
data are of class $C^\infty$, the solution $u$ is of class $C^\infty$
in $\Omega_\delta(T)$ as well).
\end{remark}

\begin{corollary}\label{Cor3.reg} Let the assumptions of Theorem
\ref{Th2.unique} hold and let $u$ be a variational solution of
problem \eqref{2.reg}. Assume, in addition, that
$$
|u(t_0,x_0)|<1
$$
for some $(t_0,x_0)\in\R^+\times\Gamma$, with $t_0>0$. Then, there exists a
neighborhood $(t_0-\delta,t_0+\delta)\times V$ of $(t_0,x_0)$ in
$\R\times\Gamma$ such that
\begin{equation}\label{3.equiv}
[\partial_n u]_{int}(t,x)=[\partial_n u]_{ext}(t,x),\ \ \forall
(t,x)\in(t_0-\delta,t_0+\delta)\times V.
\end{equation}
In particular, if
\begin{equation}\label{3.breg}
|u(t,x)|<1\ \ \text{for almost all $(t,x)\in\R^+\times\Gamma$},
\end{equation}
then the equality $[\partial_n u]_{ext}=[\partial_n u]_{int}$
holds almost everywhere in $\R^+\times\Gamma$ and, therefore,
the variational solution $u$ solves equations \eqref{2.reg} in the
usual sense.
\end{corollary}
\begin{proof} Since the solution $u$ is H\"older
continuous with respect to $t$ and $x$, there exists $\delta>0$ such that the inequality
$$
|u(t,x)|\le 1-\delta
$$
holds for all $(t,x)$ belonging to some neighborhood
$(t_0-\delta,t_0+\delta)\times V_\delta$ of $(t_0,x_0)$ in
$\R\times\Omega$. According to Proposition \ref{Prop3.h2reg}, the
sequence $u_N$ of approximate solutions (converging to the variational solution $u$) satisfies
$$
\|u_N\|_{L^\infty([t_0-\delta,t_0+\delta], H^{2}( V_\delta))}\le C,
$$
where the constant $C$ is independent of $N$. Consequently,
we can assume, without loss of generality, that $u_N\to u$
weakly-star
in this space. Thus,
$$
\partial_n u_N\big|_{\Gamma}\to\partial_n u\big|_{\Gamma}
$$
weakly in $L^2([t_0-\delta,t_0+\delta]\times V)$ (for a proper
choice of the small neighborhood $V$ of $x_0$). This convergence,
together with the definition \eqref{2.ext} of the function
$[\partial_n u]_{ext}$, give the desired equality \eqref{3.equiv}.
Thus, the first part of the statement is proved and the second one
is an immediate consequence of the first one, which finishes the proof of Corollary
\ref{Cor3.reg}.
\end{proof}

Thus, in order to prove that any variational solution $u$ of
problem \eqref{2.reg} satisfies the equations in the usual
sense, it is sufficient to check \eqref{3.breg}.
The next corollary shows that this will be the case if the
nonlinearity $f(u)$ has sufficiently strong singularities at
$\pm1$.

\begin{corollary}\label{Cor3.strong} Let the assumptions of
Theorem \ref{Th2.unique} hold and let, in addition, the potential
$F(u)$ be such that
\begin{equation}\label{3.fsing}
\lim_{u\to\pm1}F(u)=\infty.
\end{equation}
Then, for every variational solution $u$ of problem \eqref{2.reg},
$$
F(u(t))\in L^1(\Gamma) \ \text{ and }\
\|F(u(t))\|_{L^1(\Gamma)}\le C_T
$$
 for almost all $t\ge T>0$ and
condition \eqref{3.breg} holds.
\end{corollary}
\begin{proof} Let $u_N$ be a sequence of approximate
solutions converging to the variational solution $u$.
Applying estimate \eqref{A.mest} in Appendix 1 to the elliptic problem
\eqref{1.ell}, we infer that
\begin{equation}\label{3.flim}
\|F_N(u_N(t))\|_{L^1(\Gamma)}\le C_T,\ \ t\ge T,
\end{equation}
where the constant $C_T$ is independent of $N$. Using
assumption \eqref{3.fsing} and arguing as in the proof of Theorem
\ref{Th2.main}, we see that condition \eqref{3.breg} indeed holds.
Then, owing to the convergence $u_N\to u$ in $C^\gamma([0,T]\times \Omega )$, $\gamma >0$, we conclude
that $F_N(u_N)\to F(u)$ almost everywhere in $\R^+\times\Gamma$.
The Fatou lemma finally yields that $F(u(t))\in L^1(\Gamma)$, which
finishes the proof of the corollary.
\end{proof}

In particular, condition \eqref{3.breg} is satisfied if the
nonlinearity $f$ is of the form
\begin{equation}\label{3.strong}
f(u)\sim \frac {u}{(1-u^2)^p}
\end{equation}
with $p>1$. Unfortunately, the assumptions of Corollary
\ref{Cor3.strong} are violated in the physically most relevant
case of a logarithmic potential,
\begin{equation}\label{3.log}
f(u)=\ln \frac{1+u}{1-u}.
\end{equation}
Furthermore, as explained in Appendix 1, in that case, the
variational solution $u$ may indeed not be a solution in the usual
sense (even in the 1D stationary case). However, the next
proposition gives another type of sufficient condition (in terms
of the nonlinearity $g$ and the boundary external forces $h_2$)
which guarantees the equality $[\partial_n
u]_{ext}=[\partial_nu]_{int}$ and holds for the logarithmic potential
\eqref{3.log}.

\begin{proposition}\label{Prop3.boundary} Let the assumptions of
Theorem \ref{Th2.unique} hold and let, in addition, the
following inequalities hold:
\begin{equation}\label{3.bsign}
g(-1)+\eb\le h_2(x)\le g(1)-\eb,\ \ x\in\Gamma,
\end{equation}
for some $\eb>0$. Then, condition \eqref{3.breg} holds and
\begin{equation}\label{3.bfint}
\|f(u)\|_{L^1([t,t+1]\times\Gamma)} \le C_{\eb,T},\ \ t\ge T>0,
\end{equation}
where the constant $C_{\eb,T}$ is independent of the concrete
choice of the variational solution $u$. In particular, every
variational solution of \eqref{2.reg} solves this system in the
usual sense.
\end{proposition}
\begin{proof} As above, it is sufficient to derive the uniform (with respect to $N\to\infty$) estimate
\eqref{3.bfint} for the approximate solution $u_N$ of
\eqref{1.req}. In order to do so, we rewrite the system in the
elliptic-parabolic form
\begin{equation}\label{3.ep}
\begin{cases}
\Dx u_N(t)-f_N(u_N(t))-u_N(t)=\tilde h_1(t),\\
u_N\big|_{\Gamma}=\psi _N,\\
\Dt \psi _N-\Delta_\Gamma \psi _N+\partial_nu_N+g(\psi _N)=h_2.
\end{cases}
\end{equation}
Furthermore, since only the values of $g$ on the segment $[-1,1]$ are
important for the limit problem, we can assume, without loss of
generality, that
$$
g(-u)+\eb\le h_2(x)\le g(u)-\eb,\ \ u\in\R,\ \ |u|\ge1,\ \ x\in \Gamma.
$$
It follows from these inequalities and the continuity
of $g$ that
\begin{equation}\label{3.gf}
(g(z)-h_2(x)).f_N(z)\ge \frac\eb2 |f_N(z)|+C_\eb,\ \ z\in\R,\ \
x\in\Gamma,
\end{equation}
where the constant $C_\eb$ depends on $\eb$ and $g$, but is independent of
$N$.
\par
We now multiply the first equation of \eqref{3.ep} by $f_N(u_N)$
and integrate with respect to $x$. Then,  integrating by parts and
using estimate \eqref{3.gf}, we find
\begin{multline}\label{3.est7}
\frac d{dt}\int _\Gamma F_N(u_N(t))\, dS+(f'_N(u_N(t))\Nx u_N(t),\Nx
u_N(t))_\Omega+\\+(f'_N(u_N(t))\nabla_\Gamma u_N(t),\nabla_\Gamma
u_N(t))_\Gamma+\\+
1/2\|f_N(u_N(t))\|^2_{L^2(\Omega)}+\eb/2\|f_N(u_N(t))\|_{L^1(\Gamma)}\le
C(1+\|\tilde h_1(t)\|^2_{L^2(\Omega)}).
\end{multline}
Integrating this inequality with respect to time and using the facts that $f'_N\ge0$ and
the $L^2$-norm of $\tilde h_1(t)$ is controlled (see
\eqref{1.dtest} and \eqref{1.good}), we obtain
\begin{equation}\label{3.fF}
\|f_N(u_N)\|_{L^1([t,t+1]\times\Gamma)}\le \frac
2\eb(\|F_N(u_N(t))\|_{L^1(\Gamma)}+\|F_N(u_N(t+1))\|_{L^1(\Gamma)})+C_{\eb,T}.
\end{equation}
There only remains to note that the right-hand side of \eqref{3.fF}
is controlled, owing to estimate \eqref{A.mest} (exactly as in
the proof of \eqref{3.flim}). Therefore, \eqref{3.fF} gives
uniform bounds on the $L^1$-norm of $f_N(u_N)$ on the boundary.
Passing to the limit $N\to\infty$ now gives the statement of the
proposition (exactly as in the proof of Corollary \ref{Cor3.strong}).
\end{proof}

\begin{remark}\label{Rem3.negative} As already mentioned,
the variational solution $u$ may not solve equations \eqref{2.reg}
in the usual sense if conditions \eqref{3.fsing} and \eqref{3.bsign}
are violated (see Example \ref{ExA.?} for details). Furthermore,
arguing as in this example, it is not difficult to show that,
for any singular nonlinearity $f$ which does not satisfy
\eqref{3.fsing}, there exist "nonusual" variational
solutions of problem \eqref{2.reg} if the external forces $h_2$
are large enough.
\end{remark}

We conclude this section by establishing that every solution $u$ of
problem \eqref{2.reg} is separated from the singularities $\pm1$ if
the nonlinearity $f$ is singular enough. To this end, we need to require {\it at
least} condition \eqref{3.fsing} to be satisfied (see again
Example \ref{ExA.?}). Actually, we will require slightly more, namely,
that the nonlinearity $f$ satisfies the following inequalities:
\begin{equation}\label{3.ssing}
\frac{\kappa_1}{(1-u^2)^{p-1}}\le \frac{f(u)}{u}\le \frac{\kappa_2}{(1-u^2)^{M}}
\end{equation}
for some positive constants $\kappa_i$, $i=1,2$, and $M$ and where $p>2$ (recall that condition \eqref{3.fsing}
 is violated if $p<2$, so that the sufficient condition \eqref{3.ssing} is close to the necessary one). In addition,
we assume more regularity on the external forces $h_1$ and
$h_2$, namely,
\begin{equation}\label{3.boun}
h_1\in L^3(\Omega),\  \ h_2\in L^\infty(\Gamma).
\end{equation}

\begin{theorem}\label{Th3.main} Let the assumptions of Theorem
\ref{Th2.unique} hold and let, in addition, \eqref{3.ssing} and
\eqref{3.boun} be satisfied. Then, every variational solution $u$
of problem \eqref{2.reg} is separated from the singularities
$\pm1$, namely, the following estimate holds:
\begin{equation}\label{3.separated}
|u(t,x)|\le 1-\delta_T,\ \ t\ge T>0,\ \ x\in\Omega,
\end{equation}
where the constant $\delta_T$ depends on $T$, but is independent
of $u$, $t$ and $x$.
\end{theorem}
\begin{proof} We only give below the formal derivation of estimate
\eqref{3.separated} which can be justified as above by
approximating the solution $u$ by a sequence $u_N$ of solutions of
the regularized problem \eqref{1.req}. Our proof is based on the
following lemma.

\begin{lemma}\label{Lem3.p-ell} Let the assumptions of the theorem
hold and let $u(t)$ be a (variational) solution of problem
\eqref{2.reg}. Then, for every $q>0$, $f(u)\in L^q([t,t+1]\times\Omega )$ for all $t>0$
and the following estimate holds:
\begin{equation}\label{3.pest}
\|f(u)\|_{L^q([t,t+1]\times\Omega)}\le C_{T,q},\ \ t\ge T>0,
\end{equation}
where the constant $C_{T,q}$ is independent of $t$ and $u$.
\end{lemma}
\begin{proof} We rewrite system \eqref{2.reg} in the form of a
coupled elliptic-parabolic problem,
\begin{equation}\label{3.par-ell}
\begin{cases}
\Dx u-f(u)-u=\tilde h_1(t),\ \ u\big|_\Gamma=\psi ,\\
\Dt \psi -\Delta_\Gamma \psi +\partial_n u=\tilde h_2(t).
\end{cases}
\end{equation}
Then, owing to the regularity estimate \eqref{2.regular} on the
solution $u$ and conditions \eqref{3.boun}, we have
\begin{equation}\label{3.hreg}
\|\tilde h_1(t)\|_{L^3(\Omega)}+\|\tilde
h_2(t)\|_{L^\infty(\Gamma)}\le C_T,\ \ t\ge T.
\end{equation}
We introduce the function $\varphi(u):=\frac 1{1-u^2}$. Then,
$$
\varphi'(u)= \frac {2u}{(1-u^2)^2}= 2u \varphi(u)^2.
$$
We multiply the first equation of \eqref{3.par-ell} by
$u\varphi(u)^{n+1}$, where $n>1$ is an arbitrary fixed exponent,
and integrate with respect to $x\in\Omega$. Then, integrating by parts and  using the obvious
transformations
$$
(\Nx u,u\varphi'(u)|\varphi(u)|^n\Nx u)_\Omega= \frac12(|\Nx
u|^2,|\varphi'(u)|^2|\varphi(u)|^{n-2})_\Omega\ge
C_n\|\Nx(|\varphi(u)|^{n/2})\|^2_{L^2(\Omega)},
$$
$$
\int_0^u
v\varphi(v)^{n+1}\,dv=\frac12\int_0^u\varphi'(v)\varphi(v)^{n-1}\,dv=\frac1{2n}(|\varphi(v)|^n-1)
$$
and
$$
f(u)u\varphi(u)^{n+1}\ge {\kappa_1 \over 2}\varphi(u)^{n+p}-C_n
$$
(owing to the first inequality in \eqref{3.ssing}),
we have
\begin{multline}\label{3.huge}
\frac
d{dt}\|\varphi(u)\|^n_{L^n(\Gamma)}+2\kappa(\||\varphi(u)|^{n/2}\|_{H^1(\Omega)}^2+
\|\varphi(u)\|^{n+p}_{L^{n+p}(\Omega)})\le\\\le
C\|\tilde
h_1(t)\|_{L^3(\Omega)}\|\varphi(u)\|_{L^{3/2(n+1)}(\Omega)}^{n+1}+C\|\tilde
h_2(t)\|_{L^\infty(\Gamma)}\|\varphi(u)\|^{n+1}_{L^{n+1}(\Gamma)}+C
\le\\\le
C_T(\|\varphi(u)\|_{L^{3/2(n+1)}(\Omega)}^{n+1}+\|\varphi(u)\|^{n+1}_{L^{n+1}(\Gamma)}+1),
\end{multline}
where $\kappa>0$ and the constant $C_T$ depends on $n$, but is independent of $t\ge T>0$ and the
concrete choice of the solution $u$. Let us estimate the
right-hand side of \eqref{3.huge}. In order to estimate the
boundary term, we use the following trace inequality
$$
\|V\|_{L^{s+1}(\Gamma)}^{s+1}\le
C(\|V\|^2_{H^1(\Omega)}+\|V\|^{2s}_{L^{2s}(\Omega)})
$$
(which can be easily obtained by using a proper interpolation
inequality and Sobolev's embedding theorem). Using
this inequality with $V=|\varphi|^{n/2}$ and $s=1+p/n$, we find
$$
\|\varphi(u)\|^{n+p/2}_{L^{n+p/2}(\Gamma)}\le C(\||\varphi(u)|^{n/2}\|_{H^1(\Omega)}^2+
\|\varphi(u)\|^{n+p}_{L^{n+p}(\Omega)})
$$
and, therefore, since $p>2$, we can rewrite
\eqref{3.huge} without any boundary
term in the right-hand side,
\begin{multline}\label{3.huge1}
\frac
d{dt}\|\varphi(u)\|^n_{L^n(\Gamma)}+2\kappa'(\|\varphi(u)\|_{L^{3n}(\Omega)}^n+
\|\varphi(u)\|^{n+p}_{L^{n+p}(\Omega)})+\\+\kappa\|\varphi(u)\|_{L^{n+1}(\Gamma)}^{n+1}\le
C_T(\|\varphi(u)\|_{L^{3/2(n+1)}(\Omega)}^{n+1}+1)
\end{multline}
for some positive constants $\kappa'$ and $C_T$ which depend on $n$ and $T$, but
are independent of $t$ and $u$ (here, we have implicitly used the embedding $H^1\subset L^6$ and replaced
 the exponent $n+p/2$ by $n+1$ in the boundary term). In order to estimate the
right-hand side of this inequality, we use one more interpolation
inequality,
$$
\|\varphi(u)\|_{L^r(\Omega )}^s\le \|\varphi(u)\|_{L^{3n}(\Omega)}^{\theta
s}\|\varphi(u)\|_{L^{n+p}(\Omega)}^{(1-\theta)s}\le C(
\|\varphi(u)\|_{L^{3n}(\Omega)}^n+
\|\varphi(u)\|^{n+p}_{L^{n+p}(\Omega)}),
$$
where $s\in[0,1]$ and $q$ are such that
$$
\frac1r=\frac{\theta}{3n}+\frac{1-\theta}{n+p},\ \
\frac1s=\frac\theta n+\frac{1-\theta}{n+p}.
$$
Solving these equations for $r=3/2(n+1)$ and $n>p-1$, we have
$$
s=(n+1)\frac{2n-p}{2n-p-(p-2)}>n+1
$$
(since $p>2$). Thus,
$$
C_T\|\varphi(u)\|_{L^{3/2(n+1)}(\Omega)}^{n+1}\le \kappa'(\|\varphi(u)\|_{L^{3n}(\Omega)}^n+
\|\varphi(u)\|^{n+p}_{L^{n+p}(\Omega)})+C'_T
$$
and \eqref{3.huge1} yields, noting that $p>2$ and $L^{n+1}\subset L^n$,
\begin{equation}\label{3.huge2}
\frac
d{dt}\|\varphi(u(t))\|^n_{L^n(\Gamma)}+\kappa'
\|\varphi(u(t))\|^{n+2}_{L^{n+2}(\Omega)}+\kappa\|\varphi(u(t))\|_{L^{n}(\Gamma)}^{n+1}\le
C_T.
\end{equation}
Multiplying this inequality by $(t-T)^{n+1}$, we obtain
\begin{multline}\label{4.29}
\frac
d{dt}[(t-T)^{n+1}\|\varphi(u(t))\|^n_{L^n(\Gamma)}]+
\kappa'(t-T)^{n+1}\|\varphi(u(t))\|^{n+2}_{L^{n+2}(\Omega)}\le\\\le-
\kappa[(t-T)\|\varphi(u(t))\|_{L^n(\Gamma)}]^{n+1}+
C(n+1)[(t-T)\|\varphi(u(t))\|_{L^n(\Gamma)}]^n+\\+C_T(t-T)^{n+1}\le
C_{n,T}.
\end{multline}
Integrating \eqref{4.29} with respect to $t\in[T,T+2]$, we finally end up with
$$
\int_T^{T+2}(t-T)^{n+1}\|\varphi(u(t))\|^{n+2}_{L^{n+2}(\Omega)}\,dt\le
C'_{n,T}.
$$
Since $n$ is arbitrary, this last inequality, together with the second
inequality in \eqref{3.ssing}, finish the proof of
the lemma.
\end{proof}
It is now not difficult to finish the proof of the theorem. To
this end, we note that, owing to estimate \eqref{2.regular} and
Sobolev's embedding theorem,
\begin{equation}\label{3.3bound}
\|u(t)\|_{W^{2-1/3,3}(\Gamma)}\le C\|u(t)\|_{H^2(\Gamma)}\le C_T,\
\ t\ge T.
\end{equation}
On the other hand, owing to the lemma and the first equation of
\eqref{2.reg}, there holds
$$
\|\Dx u(t)\|_{L^q([t,t+1],L^3(\Omega))}\le C_T.
$$
Thus, owing to the maximal regularity for the Laplacian in $L^3$ and
Sobolev's embedding theorem,
\begin{equation}\label{3.pgrad}
\|\Nx u(t)\|_{L^q([t,t+1]\times \Omega)}\le C_{T,q},\ \  t\ge T,
\end{equation}
for any $q\ge1$. Furthermore, it follows from \eqref{2.regular}, \eqref{3.pest}
and \eqref{3.pgrad} that
\begin{equation}\label{3.finest}
\begin{aligned}
&\|\varphi(u)\|_{L^r([t,t+1], W^{1,r}(\Omega))}\le\\
&\text{\phantom{eg}}\le
C(\|\varphi(u)\|_{L^{2r}([t,t+1]\times\Omega)}+\|\varphi'(u)\|_{L^{2r}([t,t+1]\times\Omega)})(1+\|\Nx
u\|_{L^{2r}([t,t+1]\times\Omega)})\le C_{r,T},\\
&\|\Dt \varphi(u)\|_{L^{2-\eb}([t,t+1], L^{6-\eb}(\Omega))}\le
C\|\varphi(u)\|_{L^{r_\eb}([t,t+1]\times\Omega)}\|\Dt
u\|_{L^2([t,t+1]\times H^1(\Omega))}\le C_{\eb,T},
\end{aligned}
\end{equation}
where $\eb>0$ and $r>0$ are arbitrary and the constants $C_{r,T}$
and $C_{\eb,T}$ are independent of $u$ and $t\ge T$. Fixing
finally $r\gg1$ and $\eb\ll1$ in  such a way that
$$
W^{1-\eb}([t,t+1]\times\Omega)\cap L^r([t,t+1],
W^{1,r}(\Omega))\subset C([t,t+1]\times\Omega),
$$
we deduce from \eqref{3.finest} that
$$
\sup_{(s,x)\in[t,t+1]\times\Omega}\big|\frac1{1-u^2(s,x)}\big|=
\|\varphi(u)\|_{C([t,t+1]\times\Omega)}\le C_T,\ \ t\ge T,
$$
which gives \eqref{3.separated} and finishes the proof of
Theorem \ref{Th3.main}.
\end{proof}

\begin{remark}\label{Rem3.ext} It is not difficult to see that
assumption \eqref{3.boun} can be slightly relaxed, namely,
$$
h_1\in L^{r_1(p)}(\Omega),\ \ h_2\in L^{r_2(p)}(\Gamma), \ \
r_1(p)<3,\ \ r_2(p)<\infty
$$
and $r_1(p)\to3$, $r_2(p)\to\infty$ as $p\to2$ (where $p>2$ is the
exponent in inequalities \eqref{3.ssing}). It is also worth noting that the proof of
Lemma \ref{Lem3.p-ell} does not involve the Laplace-Beltrami operator
in the boundary conditions (and we have used it {\it only} in the
derivation of estimate \eqref{3.3bound}). Consequently, arguing
in a slightly more accurate way (e.g., by obtaining the
$L^\infty$-estimate on $\varphi(u)$ directly from the estimates
of Lemma \ref{Lem3.p-ell} by a Moser iteration technique), one
can extend the theorem to less regular boundary
conditions,
$$
\Dt u+\partial_n u+g(u)=h_2,\ \ x\in\Gamma.
$$
Finally, the aforementioned Moser scheme also allows to remove
the second inequality in \eqref{3.ssing} and to obtain the
separation from the singularities by only using the first
inequality in \eqref{3.ssing} for the function $f$.
We will come back to these questions elsewhere.
\end{remark}

\section{Long-time behavior: attractors and exponential
attractors}\label{s4}
In this concluding section, we study the asymptotic behavior of
 the trajectories of the solution semigroup \eqref{2.sem} acting on the
 phase space $\Phi$, endowed with the metric of $\Phi ^w$. We first recall that problem
 \eqref{1.main} enjoys the mass conservation
\begin{equation}\label{4.mass}
\<u(t)\>\equiv\<u(0)\>:=c.
\end{equation}
Therefore, it is natural to consider the restrictions of our
semigroup to the hyperplanes
\begin{equation}\label{4.const}
\Phi_c:=\{(u,\psi )\in\Phi,\ \ \<u\>=c\},\ \ c\in(-1,1),\ \
S(t):\Phi_c\to\Phi_c.
\end{equation}
The following proposition gives the existence of the global
attractor $\Cal A_c$ for this semigroup. We recall that, by
definition, a set $\Cal A_c\subset\Phi_c$ is the global attractor
for the semigroup \eqref{4.const} if
\par
1) It is compact in $\Phi_c$.
\par
2) It is strictly invariant, i.e., $S(t)\Cal A_c=\Cal A_c$, $t\ge0$.
\par
3) It attracts $\Phi_c$ as $t\to\infty$, i.e., for
every neighborhood $\Cal O(\Cal A_c)$ of $\Cal A_c$ in $\Phi_c$, there exists
$T=T(\Cal O)$ such that
$$
S(t)\Phi_c\subset\Cal O(\Cal A_c),\ \ t\ge T.
$$
We refer the reader to, e.g., \cite{BV}, \cite{handbook} and \cite{temam} for details (we note that, in our situation,
the phase space $\Phi_c$ is, by definition, bounded and, therefore, we
need not involve bounded sets in the definition of the
global attractor).

\begin{proposition}\label{Prop4.ga} Let the assumptions of Theorem
\ref{Th2.unique} hold. Then, for every $c\in(-1,1)$, the semigroup $S(t)$ associated with
the variational solutions of problem \eqref{1.main} acting on the
hyperplane \eqref{4.const} (endowed with the metric of $\Phi ^w$) possesses the global
attractor $\Cal A_c$. Furthermore, this attractor is bounded in the
space $C^\alpha(\Omega)\times C^\alpha(\Gamma)$,
$\alpha<1/4$, and is generated by all complete  trajectories
of the semigroup $S(t)$ (i.e., by all variational solutions
$(u(t),v(t))$ which are defined for all $t\in\R$).
\end{proposition}
Indeed, owing to estimate \eqref{2.slip}, the semigroup $S(t)$ has a
closed graph in $\Phi_c$. On the other hand, owing to estimate
\eqref{2.regular}, this semigroup possesses an absorbing set which is
compact in $\Phi _c$ (endowed with the metric of $\Phi ^w$) and bounded in $C^\alpha(\Omega)\times C^\alpha(\Gamma)$.
Thus, the existence of $\Cal A_c$, together with all
properties stated in the proposition, follow from a proper abstract attractor's existence theorem (see, e.g.,
\cite{BV}, \cite{handbook} and \cite{temam}).
\par
Our next task is to prove the finite-dimensionality of the global
attractor $\Cal A_c$ constructed above and the existence of
a so-called exponential attractor. We recall that, by
definition, a set $\Cal M(c)\subset\Phi_c$ is an exponential
attractor for the semigroup $S(t)$ if
\par
1) It is compact in $\Phi_c$.
\par
2) It is semiinvariant, i.e., $S(t)\Cal M(c)\subset\Cal M(c)$,
$t\ge0$.
\par
3) It has finite fractal dimension in $\Phi_c$.
\par
4) It attracts $\Phi_c$ exponentially fast as
$t\to\infty$, i.e.,
$$
\dist_{\Phi_c}(S(t)\Phi_c,\Cal M(c))\le Ce^{-\gamma t},\ \ t\ge 0,
$$
for some positive constants $\gamma$ and $C$. Here and below, $\dist_V(X,Y)$
stands for the nonsymmetric Hausdorff distance between sets in
$V$.
\par
We also recall that the usual construction of an exponential
attractor is based on the so-called squeezing (or
smoothing) property for the difference of solutions (or their
proper modifications, see \cite{EFNT}, \cite{EMZ1}, \cite{EMZ2} and \cite{handbook} for details).
The main difficulty here lies in the singular nature of
the equations at $\pm1$. In particular, the difference between two singular solutions $u_1$
and $u_2$ does not possess any regularization. Nevertheless, as we
will see below, our nonlinearity is strictly monotone near the
singularities $\pm1$ and, far from these singularities, the problem still
possesses the usual parabolic smoothing property. This fact,
together with the H\"older continuity  of the solutions and some
localization technique, allow to construct an exponential
attractor by using a proper modification of the techniques developed
in \cite{EMZ1} and \cite{EZ1}. However, some additional difficulties arise here, due
to the fact that the $H^{-1}$-norm is not local.

\begin{theorem}\label{Th4.main} Let the assumptions of Theorem \ref{Th2.unique}
hold.
Then, the semigroup $S(t)$ acting on the phase space $\Phi_c$ (endowed
with the metric of $\Phi $)
possesses an exponential attractor $\Cal M(c)$ which is bounded in
$C^\alpha(\Omega)\times C^\alpha(\Gamma)$, $\alpha <1/4$.
\end{theorem}
\begin{proof}
  We first note that, owing to Theorem \ref{Th2.main},
there exists a compact (for the metric of $\Phi^w$) absorbing set
\begin{equation}\label{4.abs}
\Bbb B_c:=S(1)\Phi_c
\end{equation}
such that
\begin{equation}\label{4.inv}
S(t)\Bbb B_c\subset\Bbb B_c
\end{equation}
and
\begin{multline}\label{4.good}
\|u\|_{C^\alpha([t,t+1]\times\Omega)}+\|u(t)\|_{H^2(\Gamma)}+
\|\Dt u(t)\|_{H^{-1}(\Omega)}+\|\Dt u(t)\|_{L^2(\Gamma)}+\\+
\|f(u(t))\|_{L^1(\Omega)}+\|\Dt
u\|_{L^2([t,t+1],H^1(\Omega))}+\|\Dt
u\|_{L^2([t,t+1],H^1(\Gamma))}\le R
\end{multline}
for some fixed constant $R$ which depends on $c$, but is independent of
$(u(0),\psi (0))\in\Bbb B_c$. In particular, for every point
$(u,\psi )\in\Bbb B_c$, we have
$$
\psi =u\big|_{\Gamma}
$$
and, consequently, we generally write $u$ instead of $\psi $
in the boundary norms.
\par
Thus, we only need to construct an exponential attractor $\Cal M(c)$
for the semigroup $S(t)$ restricted to the semiinvariant absorbing
set $\Bbb B_c$. As usual, to do so, we need to obtain
proper estimates on the difference of two solutions $u_1(t)$ and
$u_2(t)$ starting from the set $\Bbb B_c$. Furthermore, we use the
following natural norm on the phase space $\Bbb B_c$:
\begin{multline}\label{4.norm}
\|u_1-u_2\|_{\Phi^w}^2:=\|u_1-u_2\|^2_{H^{-1}(\Omega)}+\|u_1-u_2\|^2_{L^2(\Gamma)}=\\=
\|(-\Dx)^{-1/2}(u_1-u_2)\|_{L^2(\Omega)}^2+
\|u_1-u_2\|^2_{L^2(\Gamma)}.
\end{multline}
A crucial point in the proof is
the global Lipschitz continuity in this norm (see Theorem 2.2),
\begin{multline}\label{4.lip}
\|u_1(t)-u_2(t)\|^2_{\Phi^w}+\\+\int_t^{t+1}(\|u_1(s)-u_2(s)\|_{H^1(\Omega)}^2+\|u_1(s)-u_2(s)\|_{H^1(\Gamma)}^2)\,ds\le
Ce^{Kt}\|u_1(0)-u_2(0)\|_{\Phi^w}^2,
\end{multline}
where the positive constants $C$ and $K$ are independent of
$u_1(0),u_2(0)\in\Bbb B_c$.
\par
We now consider an arbitrary small $\eb$-ball $B(\eb,u_0,\Phi^w)$ in the space
$\Bbb B_c$ (endowed with the metric of $\Phi^w$) and centered at
$u_0$, where $0<\eb\le\eb_0\ll1$ (and the parameter
$\eb_0$ will be fixed below). Let also $u^0(t)$, $t\ge0$, be the solution
of problem \eqref{2.reg} starting from $u_0$.
\par
As in \eqref{3.odelta}, we introduce the sets
\begin{equation}\label{4.odelta}
\Omega_\delta(u_0):=\{x\in\Omega,\ |u_0(x)|< 1-\delta\},\ \
\overline{\Omega}_\delta(u_0):=\{x\in\Omega,\ \ |u_0(x)|>1-\delta\},
\end{equation}
where $\delta$ is a sufficiently small positive number. Then,
since the function $u_0(x)$ is uniformly H\"older continuous in
$\Omega$, there holds
\begin{equation}\label{4.dist}
d(\partial\Omega_{\delta_1}(u_0),\partial\Omega_{\delta_2}(u_0))\ge
C_{\delta_1,\delta_2}>0,\ \ \delta_1\ne\delta_2,
\end{equation}
where the constant $C_{\delta_1,\delta_2}$ depends on $\delta_i$, $i=1,2$,
but is independent of the concrete choice of $u_0\in\Bbb B_c$.
\par
As a next step, we note that, owing to the uniform H\"older
continuity of the trajectory $u^0(t)$ (in space and time), there
exists $T=T(\delta)$ such that
\begin{equation}
\begin{aligned}
|u^0(t)|\le 1-\frac\delta2,\ x\in\Omega_\delta(u_0),\ t\in[0,T],\\
|u^0(t)|\ge1-2\delta,\ x\in\overline{\Omega}_{2\delta}(u_0),\ t\in[0,T],
\end{aligned}
\end{equation}
and, furthermore, owing again to the uniform H\"older continuity,
$$
\|u_1(t)-u_2(t)\|_{C(\Omega)}\le
C\|u_1(t)-u_2(t)\|_{\Phi^w}^\kappa\|u_1(t)-u_2(t)\|^{1-\kappa}_{C^\alpha(\Omega )}\le
C_T\eb^{\kappa},
$$
for all $u_1(0),u_2(0)\in B(\eb,u_0,\Phi^w)$. We can thus
fix $\eb_0=\eb_0(\delta)$ in such a way that
\begin{equation}\label{4.sep}
\begin{aligned}
|u(t)|\le 1-\frac\delta 4,\ \  x\in\Omega_\delta(u_0),\ \ t\in[0,T],\\
|u(t)|\ge1-4\delta,\ \ x\in\overline{\Omega}_{2\delta}(u_0),\ \ t\in[0,T],
\end{aligned}
\end{equation}
for all trajectories $u(t)$ starting from the ball
$B(\eb,u_0,\Phi^w)$ with $\eb\le\eb_0$.
\par
We also introduce the cut-off function
 $\theta\in C^\infty(\R^3,[0,1])$ such that
\begin{equation}\label{4.cut}
\theta(x)\equiv0,\ \ x\in \overline{\Omega}_\delta(u_0),\ \
\theta(x)\equiv1,\ \ x\in \Omega_{2\delta}(u_0).
\end{equation}
Such a function exists, owing to condition \eqref{4.dist}. Furthermore,
it follows from this condition that this function can be chosen in
such a way that it satisfies
the additional assumption
\begin{equation}\label{4.phi}
\|\theta\|_{C^k(\R^3)}\le C_k,
\end{equation}
where $k\in\Bbb N$ is arbitrary and
the constant $C_k$  depends on $\delta$, but is
independent of the choice of $u_0\in\Bbb B_c$, see \cite{EZ1} for
details.
\par
Finally the second estimate of \eqref{4.sep} yields
\begin{equation}\label{4.mon}
f'(u(t,x))\ge\Lambda(\delta),\ \ x\in\overline{\Omega}_{2\delta}(u_0),\ \
t\in[0,T],
\end{equation}
for all trajectories $u(t)$ starting from the ball
$B(\eb,u_0,\Phi^w)$, where
$$
\Lambda(\delta):=\min\{f'(1-4\delta),f'(-1+4\delta)\}.
$$
Since $f'(u)\to\infty$ as $u\to\pm1$, then $\Lambda(\delta)\to\infty$ as
$\delta\to0$ and we can fix $\delta>0$ in such a way that $\Lambda(\delta)$
is arbitrarily large. This will be essentially used in the
next lemma which gives some kind of smoothing  property for the difference
of two solutions $u_1$ and $u_2$ and is crucial for our construction.

\begin{lemma}\label{Lem4.contr} Let the above assumptions hold.
Then, there exists $\delta>0$ such that the following estimate holds:
\begin{multline}\label{4.contr}
\|u_1(T)-u_2(T)\|_{\Phi^w}^2\le e^{-\beta T}\|u_1(0)-u_2(0)\|_{\Phi^w}^2+\\+C\int_0^T\|\theta
(u_1(s)-u_2(s))\|^2_{L^2(\Omega)}\,ds,
\end{multline}
where the positive constants $\delta$, $\beta$ and $C$ are
independent of $u_1,u_2\in B(\eb,u_0,\Phi^w)$, $s$ and $u_0\in\Bbb B_c$.
\end{lemma}
\begin{proof} As usual, we only give the formal derivation of this
estimate which can be justified by approximating the variational
solutions $u_1(t)$ and $u_2(t)$ by appropriate solutions of
the regular equation \eqref{1.req}. Set $v(t):=u_1(t)-u_2(t)$.
Then, this function (formally) solves
\begin{equation}\label{4.4order}
\begin{cases}
\Dt v=-\Dx (\Dx v-l(t)v+\lambda v),\ \ \partial_n(\Dx
v-l(t)v+\lambda v)\big|_{\Gamma}=0,\\
\Dt v-\Delta_\Gamma v+\partial_n v+m(t)v=0,\ \ \text{on $\Gamma$},
\end{cases}
\end{equation}
where
$$
l(t):=\int_0^1f'(s u_1(t)+(1-s)u_2(t))\,ds,\ \
m(t):=\int_0^1g'(su_1(t)+(1-s)u_2(t))\,ds.
$$
Multiplying this equation by $(-\Dx)^{-1}v(t)$, integrating over $\Omega$ and using the fact that
$\<v(t)\>\equiv0$, we obtain
\begin{equation}\label{4.main}
\frac12\frac d{dt}\|v(t)\|^2_{\Phi^w}+\|\Nx
v(t)\|^2_{L^2(\Omega)}+(l(t)v(t),v(t))_\Omega\le\lambda\|v(t)\|^2_{L^2(\Omega )}+K\|v(t)\|_{L^2(\Gamma)}^2,
\end{equation}
where $K=\|g'\|_{C([-1,1])}$. We estimate the most complicated
term $(l(t)v,v)$ as follows:
\begin{multline}\label{4.l}
\int_\Omega
l(t,x)|v(x)|^2\,dx\ge
\int_{\overline{\Omega}_{2\delta}(u_0)}l(t,x)|v(x)|^2\,dx\ge\\\ge
\Lambda\|v\|^2_{L^2(\Omega)}-\Lambda\|v\|_{L^2(\Omega_{2\delta}(u_0))}^2\ge
\Lambda\|v\|^2_{L^2(\Omega)}-\Lambda\|\theta v\|^2_{L^2(\Omega)}.
\end{multline}
Thus, inequality \eqref{4.main} reads
\begin{equation}\label{4.main1}
\frac12\frac d{dt}\|v(t)\|^2_{\Phi^w}+\|\Nx
v(t)\|^2_{L^2(\Omega)}+(\Lambda-\lambda)\|v(t)\|^2_{L^2(\Omega)}\le K\|v(t)\|_{L^2(\Gamma)}^2+
\Lambda\|\theta v(t)\|_{L^2(\Omega)}^2.
\end{equation}
Furthermore, using the trace-interpolation estimate
$$
\|v\|_{L^2(\Gamma)}^2\le
C\|v\|_{H^1(\Omega)}\|v\|_{L^2(\Omega)}\le
C(\Lambda-\lambda)^{-1/2}(\|v\|_{H^1(\Omega)}^2+(\Lambda-\lambda)\|v\|^2_{L^2(\Omega)})
$$
and fixing $\delta$ in such a way that $KC\le
1/2(\Lambda(\delta)-\lambda)^{1/2}$, we finally end up with
\begin{equation}\label{4.main2}
\frac d {dt}\|v(t)\|^2_{\Phi^w}+\|\Nx
v(t)\|^2_{L^2(\Omega)}+\beta(\|v(t)\|^2_{L^2(\Omega)}+\|v(t)\|^2_{L^2(\Gamma)})\le
2\Lambda\|\theta v(t)\|_{L^2(\Omega)}^2,
\end{equation}
where $\beta>0$. Using the Poincar\'e inequality
$\|v\|_{H^{-1}(\Omega )}\le C\|v\|_{L^2(\Omega )}$, together with the Gronwall
inequality, we deduce \eqref{4.contr} and finish the proof of the
lemma.
\end{proof}
Thus, owing to Lemma \ref{Lem4.contr}, the semigroup $S(t)$ is a
contraction, up to the term
$\|\theta(u_1-u_2)\|_{L^2([0,T]\times\Omega)}$. The next lemma gives
some kind of compactness for this term.

\begin{lemma}\label{Lem4.comp} Let the above assumptions hold.
Then, the following estimate holds:
\begin{multline}\label{4.comp}
\|\Dt(\theta(u_1-u_2))\|_{L^2([0,T],H^{-3}(\Omega))}+\\+\|\theta(u_1-u_2)\|_{L^2([0,T],H^1(\Omega))}\le
Ce^{KT}\|u_1(0)-u_2(0)\|_{\Phi^w},
\end{multline}
where the constants $C$ and $K$ are independent of $u_i(0)\in
B(\eb,u_0,\Phi^w)$, $i=1,2$, and $u_0\in\Bbb B_c$.
\end{lemma}
\begin{proof} The second term in the left-hand side of
\eqref{4.comp} can be easily estimated by \eqref{4.lip} (and the
fact that $\Nx\theta$ is uniformly bounded). So, we only need to
estimate the time derivative. To this end, we recall that $\Dt v$ ($v=u_1-u_2$)
satisfies
$$
\Dt v=-\Dx(\Dx v-l(t)v)
$$
in the sense of distributions. Therefore, for any test function
$\varphi\in C^\infty_0(\Omega)$, there holds
\begin{multline}
\<\Dt (\theta v(t)),\varphi\>_\Omega=
-\<\Dx v(t)-l(t)v(t),\Dx (\theta\varphi)\>_\Omega=\\=
\<\Nx v(t),\Nx\Dx(\theta\varphi)\>_\Omega+
\<l(t)v(t),\Dx(\theta\varphi)\>_\Omega.
\nonumber
\end{multline}
Since $\supp\theta\subset\Omega_{\delta}(u_0)$, \eqref{4.sep} yields
$$
|\<l(t)v,\Dx(\theta\varphi)\>_\Omega|\le
C\|v\|_{L^2(\Omega)}\|\varphi\|_{H^2(\Omega)}
$$
and, thus,
$$
|\<\Dt(\theta v(t)),\varphi\>_\Omega|\le
C_1\|v(t)\|_{H^1(\Omega)}\|\varphi\|_{H^3(\Omega)}.
$$
This estimate, together with \eqref{4.lip}, give the desired
estimate \eqref{4.comp} on the time derivative and finish the
proof of the lemma.
\end{proof}
It is now not difficult to finish the proof of the theorem.
We introduce the functional spaces
\begin{equation}
\begin{aligned}
&\Bbb H_1:=L^2([0,T],H^1(\Omega))\cap H^1([0,T],H^{-3}(\Omega)),\\
&\Bbb H:=L^2([0,T],L^2(\Omega)).
\end{aligned}
\end{equation}
Then, obviously, $\Bbb H_1$ is compactly embedded into $\Bbb H$. We also introduce, for any $u_0\in\Bbb B_c$, the linear operator
$$
\Bbb K_{u_0}:B(\eb,u_0,\Phi^w)\to\Bbb H_1
$$
by
$$
\Bbb K_{u_0}u(0):=\theta u(\cdot), \text{ $u(t)$ solves
\eqref{2.reg}}
$$
(where the constants $\delta$, $T$ and the cut-off function $\theta$
are such that Lemmas \ref{Lem4.contr} and \ref{Lem4.comp} hold).
Then, on the one hand, owing to Lemma \ref{Lem4.comp}, the map $\Bbb
K_{u_0}$ is uniformly Lipschitz continuous,
\begin{equation}\label{4.lipcont}
\|\Bbb K_{u_0}(u_1-u_2)\|_{\Bbb H_1}\le L\|u_1-u_2\|_{\Phi^w},\ \
u_1,u_2\in B(\eb,u_0,\Phi^w),\ \ \eb\le\eb_0,
\end{equation}
where the Lipschitz constant $L$ is independent of the choice of
$u_0\in\Bbb B_c$ and $\eb\le\eb_0$. On the other hand, it follows from Lemma \ref{Lem4.contr} that
\begin{equation}\label{4.contrs}
\|S(T)u_1-S(T)u_2\|_{\Phi^w}\le
(1-\gamma)\|u_1-u_2\|_{\Phi^w}+C\|\Bbb K(u_1-u_2)\|_{\Bbb H},
\end{equation}
where $\gamma>0$ and $C>0$ are also independent of $u_0\in \Bbb
B_c$, $\eb\le\eb_0$ and $u_1,u_2\in B(\eb,u_0,\Phi^w)$.
\par
It is known (see, e.g., \cite{EZ2}; see also \cite{MP}) that inequalities
\eqref{4.lipcont} and \eqref{4.contrs}, together with the
compactness of the embedding $\Bbb H_1\subset\Bbb H$, guarantee
the existence of an exponential attractor $\Cal M_d(c)\subset\Bbb B_c$
for the discrete semigroup $S(nT)$ acting on the phase space $\Bbb
B_c$ (endowed with the topology of $\Phi^w$). Furthermore,
\eqref{4.lip}, together with the control \eqref{4.good} of the time
derivative, yield that the semigroup $S(t)$ is uniformly H\"older continuous
with respect to time and space in $[0,T]\times\Bbb B_c$. Thus, the
desired exponential attractor $\Cal M(c)$ for the continuous
semigroup $S(t)$ on $\Bbb B_c$ can be obtained by the standard
formula
$$
\Cal M(c):=\cup_{t\in[0,T]}\Cal M_d(c).
$$
Finally, although we have formally constructed the exponential
attractor $\Cal M(c)\subset \Bbb B_c\subset C^\alpha(\Omega)\times C^\alpha(\Gamma)$
in the topology of $\Phi^w$ only, the
control of the $C^\alpha$-norm of $\Bbb B_c$, together with a
proper interpolation inequality, give the finite-dimensionality and
the exponential attraction in the initial topology of $\Phi_c$ as
well. This finishes the proof of Theorem \ref{Th4.main}.
\end{proof}

\section{Appendix 1. Some auxiliary results}\label{sA1}
In this section, we establish several estimates which are used
in the paper. We start with regularity results for the
following singular elliptic boundary value problem:
\begin{equation}\label{A.1}
\begin{cases}
\Dx u-u-f(u)=\tilde h_1,\ \ {\rm in}\ \Omega,\\
\partial_n u+u-\Delta_\Gamma u=\tilde h_2, \ \
{\rm on}\ \Gamma,
\end{cases}
\end{equation}
where $\tilde h_1\in L^2(\Omega)$, $\tilde h_2\in L^2(\Gamma)$ and
the nonlinearity $f$ satisfies conditions \eqref{A.2}. As above,
a solution $u$ of this problem should be understood as a
variational solution, analogously to  Definition \ref{Def2.var}.
Therefore, in order to justify the estimates given below, we factually  need to deduce the corresponding
uniform estimates for regularized problems of the form \eqref{A.1}, where $f$ is replaced
by its approximations $f_N$ (defined by \eqref{1.f}), and then pass
to the limit $N\to\infty$. Since this passage to the limit is
explained in details in Section \ref{s2}, we give below the
formal derivation of these estimates directly for the limit
singular problem \eqref{A.1}, leaving the justifications to the
reader.

\begin{theorem}\label{ThA.1} Let the above assumptions hold.
Then, the following estimate holds for the solution $u$ of problem \eqref{A.1}:
\begin{equation}\label{A.7}
\|u\|_{H^{1}(\Omega)}^2+\|u\|_{H^{1}(\Gamma)}^2+\|f(u)\|_{L^1(\Omega)}\le
C(1+\|\tilde h_1\|_{L^2(\Omega)}^2+\|\tilde h_2\|_{L^2(\Gamma)}^2),
\end{equation}
where the constant $C$ is
independent of $\tilde h_1$ and $\tilde h_2$. Furthermore, $u\in
C^\alpha(\Omega)\cap H^{2}(\Gamma )$ with $\alpha<1/4$
and the following estimate holds:
\begin{equation}\label{A.m}
\|u\|_{C^\alpha(\Omega)}^2+\|u\|_{H^{2}(\Gamma)}^2\le
C(1+\|\tilde h_1\|_{L^2(\Omega)}^2+\|\tilde h_2\|_{L^2(\Gamma)}^2).
\end{equation}
Finally, $F(u)\in L^1(\Gamma)$, where $F(z):=\int_0^zf(s)\,ds$,
$\Nx D_\tau u\in L^2(\Omega)$,
$u\in H^{2}(\Omega_\eb)$, for every $\eb>0$, where $\Omega_\eb:=\{x\in\Omega,\
d(x,\Omega)>\eb\}$, and the following estimate holds:
\begin{equation}\label{A.mest}
\|F(u)\|_{L^1(\Gamma)}+\|u\|_{H^{2}(\Omega_\eb)}^2+\|\Nx D_\tau
u\|^2_{L^2(\Omega)}\le C_\eb(\|\tilde
h_1\|^2_{L^2(\Omega)}+\|\tilde h_2\|^2_{L^2(\Gamma)}).
\end{equation}
\end{theorem}
\begin{proof} Estimate \eqref{A.7} can be obtained by
multiplying the first equation by $u$ and integrating over
$\Omega$ (note that the existence and uniqueness of the solution $u$ can be obtained exactly
as in Section \ref{s2}). So, we only need to give a formal
derivation of estimates \eqref{A.m} and \eqref{A.mest}.
\par
  The derivation of these estimates is based on
a standard localization technique. Thus, we only give below a
sketch of the proof, leaving the details to the reader. Let $\theta$
be a smooth nonnegative cut-off function such that $\theta(x)=1$ if
$d(x,\Gamma)\ge \eb$ and $\theta(x)=0$ if
$d(x,\Gamma)\le\eb/2$ which satisfies, in addition, the inequality
\begin{equation}\label{A.14}
|\Nx\theta(x)|\le C\theta^{1/2}(x).
\end{equation}
Then, multiplying equation
\eqref{A.1} by
$$
\sum_{i=1}^3\partial_{x_i}(\theta(x)\partial_{x_i} u),
$$
integrating by parts and using estimate \eqref{A.7} (in order to
estimate the lower-order terms) and the fact that $f'\ge0$,
we deduce that
\begin{equation}\label{A.15}
\|u\|_{H^{2}(\Omega_\eb)}^2\le
C(1+\|\tilde h_1\|_{L^2(\Omega)}^2+\|\tilde
h_2\|_{L^2(\Gamma)}^2),
\end{equation}
where the constant $C=C_\eb$ depends on $\eb>0$, but is
independent of $u$, $\tilde h_1$ and $\tilde h_2$.
\par
 Since $H^{2}\subset
C^\alpha$, $\alpha<1/2$, there only remains, in order to finish the proof of the theorem, to study the
function $u$ in a small $\eb$-neighborhood of the boundary
$\Gamma$.
\par
Let $x_0\in\Gamma$ and $y=y(x)$ be local coordinates
in the neighborhood of $x_0$ such that $y(x_0)=0$ and $\Omega$ is defined in these
coordinates by the condition $y_1>0$. Then, in the variable $y$,
problem \eqref{A.1} reads
\begin{equation}\label{A.16}
\begin{cases}
\sum_{i,j=1}^3\partial_{y_i}(a_{ij}(y)\partial_{y_j}u)+\sum_{i=1}^3b_i(y)\partial_{y_i}u+c(y)u-f(u)=\tilde h_1,\ \ y_1>0,\\
\sum_{i,j=2}^3\partial_{y_i}(d_{ij}(y)\partial_{y_j}u)+\sum_{i=2}^3e_i(y)\partial_{y_i}u+g(y)u+\tilde h_2=\partial_{y_1}u,\
\ y_1=0,
\end{cases}
\end{equation}
where $a_{ij}$, $b_i$, $c$, $d_{ij}$, $e_i$ and $g$ are
smooth functions which satisfy uniform ellipticity assumptions.
\par
Differentiating the first equation of \eqref{A.16} with respect to
${y_k}$, $k=2,3$, multiplying the resulting equation by $\phi v_k$, where
$v_k:=\partial_{y_k}u$ and $\phi$ is a smooth nonnegative cut-off function
which is equal to one in the ball $|y|\le\eb$ and zero outside the
ball $|y|\ge2\eb$ and satisfies \eqref{A.14}, using again the fact
that $f'\ge0$ and the ellipticity assumption on the $a_{ij}$, we find, after
standard transformations,
\begin{equation}\label{A.18}
\gamma(\phi|\Nx v_k|,|\Nx v_k|)_\Omega +(\phi
v_k,a_{11}(y)\partial_{y_1}v_k)_{\Gamma}+(\phi
v_k,v_k)_\Omega \le C(\|u\|_{H^{1}(\Omega)}^2+\|\tilde
h_1\|_{L^2(\Omega)}^2),
\end{equation}
where the positive  constants $C$ and $\gamma$ are independent of  $u$.
 Differentiating then the second equation of
\eqref{A.16} with respect to $y_k$, inserting the expression for
$\partial_{y_1}v_k$ thus obtained into \eqref{A.18} and arguing
analogously, we have
\begin{multline}\label{A.19}
\gamma(\phi|\Nx v_k|,|\Nx v_k|)_\Omega +\gamma
(\phi|\nabla_{y_2,y_3}v_k|,|\nabla_{y_2,y_3}v_k|)_{\Gamma}
+(\phi v_k,v_k)_{\Omega}+(\phi
v_k,v_k)_\Gamma\le\\\le C(\|u\|_{H^{1}(\Omega)}^2+\|\tilde
h_1\|_{L^2(\Omega)}^2+\|u\|_{H^{1}(\Gamma)}^2+\|\tilde h_2\|_{L^2(\Gamma)}^2).
\end{multline}
Combining this estimate with \eqref{A.7}, we finally end up with
\begin{multline}\label{A.20}
\|\phi u\|_{L^2(\R^+_{y_1},H^{2}(\R^2_{y_2,y_3}))}^2+
\|\phi u\|^2_{H^{1}(\R^+_{y_1},H^{1}(\R^2_{y_2,y_3}))}+
\|\phi u\big|_{y_1=0}\|_{H^{2}(\R^2_{y_2,y_3})}^2\le\\\le
 C(1+\|\tilde h_1\|_{L^2(\Omega)}^2+\|\tilde h_2\|_{L^2(\Gamma)}^2),
\end{multline}
where the constant $C$ is independent
of $x_0$ and $u$.
\par
Returning to the variable $x$ and using the fact that the boundary
point $x_0$ is arbitrary (and that $\Gamma$ is smooth),
we infer from \eqref{A.20} that
\begin{equation}\label{A.2g}
\|\Nx D_\tau u\|^2_{L^2(\Omega)}+\|u\|_{H^{2}(\Gamma)}^2\le
C(1+\|\tilde h_1\|_{L^2(\Omega)}^2+\|\tilde
h_2\|_{L^2(\Gamma)}^2).
\end{equation}
In addition,
owing to the embedding
$$
L^2(\R,H^{2}(\R^2))\cap H^{1}(\R,H^{1}(\R^2))\subset
C^\alpha(\R^3),\ \ \alpha<1/4,
$$
estimate \eqref{A.20}, together with \eqref{A.15}, also imply the estimate
$$
\|u\|_{C^\alpha(\Omega)}^2\le C(1+\|\tilde h_1\|_{L^2(\Omega)}^2+\|\tilde
h_2\|_{L^2(\Gamma)}^2),\ \ \alpha<1/4.
$$
Thus, in order to finish the proof of the theorem, we only need to
estimate the $L^1$-norm of $F(u)$ on the boundary. To this end, we
also use the localized equations \eqref{A.16}, but now
multiply the first one by $\phi\partial_{y_1} u$. Then, after
obvious transformations, we have
\begin{multline}
\partial_{y_1}(\frac12\phi(y) a_{11}(y)|\partial_{y_1}
u|^2+\phi(y)F(u))\ge\\\ge -C(\phi+|\nabla_y\phi|+|D^2_y\phi|)(|\nabla_y u|^2+|\partial_{y_1,y_2}^2
u|^2+|\partial^2_{y_1,y_3} u|^2+|D_{y_2,y_3}^2u|^2+F(u)+|\tilde h_1|^2),
\end{multline}
where the constant $C$ is independent of $x_0\in\Gamma$ and $u$.
 Integrating this estimate with respect to
$y\in\R^+\times\R^2$ and using \eqref{A.7} and \eqref{A.20},
together with the fact that $\phi\ne0$ only in a small
neighborhood of the boundary, we see that
\begin{multline}\label{A.strgange}
\int_{\R^2}(\phi(0,y_1,y_2)F(u(0,y_1,y_2))-\phi(0,y_1,y_2)a_{11}(0,y_1,y_2)|\partial_{y_1}
u(0,y_1,y_2)|^2)\,dy_1\,dy_2\le\\\le
C(1+\|\tilde h_1\|^2_{L^2(\Omega)}+\|\tilde
h_2\|^2_{L^2(\Gamma)}).
\end{multline}
Thus, keeping in mind the fact that $\phi(y)$ is nonnegative and is equal to
one close to the given point $x_0\in\Gamma$, we conclude, returning to the variable $x$, that
$$
\|F(u)\|_{L^1(\Gamma)}\le C\|\partial_n
u\|_{L^2(\Gamma)}^2+C(1+\|\tilde h_1\|^2_{L^2(\Omega)}+\|\tilde
h_2\|^2_{L^2(\Gamma)}).
$$
There now only remains to note that, owing to estimate \eqref{A.2g} and
the second equation of \eqref{A.1}, we can control the $L^2$-norm of
$\partial_n u$ on the boundary,
$$
\|\partial_n u\|^2_{L^2(\Gamma)}\le C(1+\|\tilde h_1\|^2_{L^2(\Omega)}+\|\tilde
h_2\|^2_{L^2(\Gamma)}).
$$
The control of the $L^1$-norm of $F(u)$ on the boundary thus follows and Theorem \ref{ThA.1} is proved.
\end{proof}
Our next task is to give an example of equations \eqref{A.1} for which the
solution $u$ does not satisfy the equations in the usual sense, but
only in the variational sense described in Section \ref{s2}.
We recall that such an example cannot be found if the potential $F(u)$
is singular at $\pm1$.
However, the situation is essentially different if the potential
$F(u)$ has finite limits as $u\to\pm1$.
Indeed, in that case,  the control of the $L^1$-norm of $F(u)$ is of
no use and
 the singular part of the boundary (where $|u(x)|=1$) may now
have positive measure and may even coincide with the whole
boundary. As we can see from the following example, the equality $[\partial_n u]_{int}=[\partial_n u]_{ext}$
can be violated at such singular points.

\begin{example}\label{ExA.?}
We consider the following  example of a one dimensional boundary value problem
of the form \eqref{A.1}:
\begin{equation}\label{A.ODE}
y''-f(y)=0, \ \  y'(\pm1)= K\ge0,\ \ x\in[-1,1],
\end{equation}
where the function $f$ satisfies assumptions \eqref{A.2} and, in
addition, $F(1)=F_1<\infty$ and $f(-y)=-f(y)$,
which is of course a particular case of our general theory. Then, an analysis of the above ODE shows that, for relatively small values of $K$, this problem has a regular usual
 solution $y_K(x)$ which is odd,
$$
y_K(-x)=-y_K(x)
$$
(owing to the symmetry and the uniqueness), and is separated from the singularities of $f$. However, there exists a critical value $K_+$ such that, for $K>K_+$, $y_K$ coincides with the singular solution $y_+$ of the problem
$$
y''_+-f(y_+)=0, \ \ y_+(1)=1,\ \ y_{+}(-1)=-1.
$$

\noindent Thus, the usual solution of
\eqref{A.ODE} does not exist for $K>K_+$. However, for these values of $K$,
 it can be uniquely defined   as a variational solution. For the reader's convenience, we also give below a simple
alternative proof of the above nonexistence fact which can be partially
extended to the multi-dimensional case. Since $\|y_K\|_{L^\infty([-1,1])}\le 1$, the usual
interior regularity techniques (see the interior regularity estimate
in Theorem \ref{ThA.1}) show that
\begin{equation}\label{A.int}
|y'_K(x)|\le C,\ \ |y_K(x)|\le 1-\delta,\ \ x\in(-1/2,1/2),
\end{equation}
where the positive constants $C$ and $\delta$ are  {\it independent } of $K$. Multiplying
now equation \eqref{A.ODE} by $y'$, integrating over $[0,1]$
and using \eqref{A.int}, we obtain
\begin{equation}\label{A.str}
|\frac12|y_K'(1)|^2-F(y_K(1))|\le C,
\end{equation}
where the constant $C$ is again independent of $K$. Thus, $y_K$
cannot satisfy the boundary condition $y_K'(1)=K$ if $K$ is
large enough and $F(1)$ is finite.
\end{example}

\begin{remark}\label{RemA.equ} Let $y_K(x)$ be a variational
solution of problem \eqref{A.ODE} as constructed in the previous
section. Then, since this solution is odd, it automatically
satisfies the equation
$$
y''-f(y)=\<y''-f(y)\>_{[-1,1]}
$$
and, therefore, it is a (variational) equilibrium for the
corresponding 1D Cahn-Hilliard problem of the form \eqref{1.main}.
Thus, even in the 1D case, problem \eqref{1.main} can have variational
solutions which do not satisfy the boundary conditions in the usual
sense.
\end{remark}

\end{document}